\documentclass{article}
\usepackage{fullpage}

\usepackage{mathtools, amssymb, amsthm}
\usepackage{xspace}
\usepackage{hyperref}
\usepackage{cleveref}
\usepackage{subcaption}
\usepackage{graphicx}

\usepackage{easyReview}

\newtheorem{theorem}{Theorem}
\newtheorem{definition}{Definition}
\newtheorem{lemma}{Lemma}

\theoremstyle{plain}
\newtheorem{redrule}{Reduction Rule}
\Crefname{redrule}{Reduction Rule}{Reduction Rules}

\newcommand{\NP}{\ensuremath{\mathsf{NP}}\xspace}

\renewcommand{\P}{\ensuremath{\mathsf{P}}\xspace}

\newcommand{\containment}{\ensuremath{\mathsf{NP \subseteq coNP/poly}}\xspace}

\newcommand{\N}{\mathbb{N}}
\newcommand{\Z}{\mathbb{Z}}

\newcommand{\Oh}{\mathcal{O}}

\newcommand{\cQ}{\ensuremath{\mathcal{Q}}}

\newcommand{\cC}{\ensuremath{\mathcal{C}}}
\newcommand{\cD}{\ensuremath{\mathcal{D}}}
\newcommand{\cI}{\ensuremath{\mathcal{I}}}
\newcommand{\cJ}{\ensuremath{\mathcal{J}}}
\newcommand{\cF}{\ensuremath{\mathcal{F}}}

\newcommand{\cX}{\ensuremath{\mathcal{X}}}

\newcommand{\oY}{\ensuremath{\overline{Y}}}
\newcommand{\oX}{\ensuremath{\overline{X}}}
\newcommand{\oS}{\ensuremath{\overline{S}}}
\newcommand{\oQ}{\ensuremath{\overline{Q}}}
\newcommand{\oG}{\ensuremath{\overline{G}}}

\newcommand{\graphclass}[1]{\ensuremath{\cC^{\texttt{#1}}}}
\newcommand{\Cindependent}{\graphclass{independent}}
\newcommand{\Ccomplete}{\graphclass{complete}}
\newcommand{\Ctree}{\graphclass{tree}}
\newcommand{\Cforest}{\graphclass{forest}}
\newcommand{\Ctdd}{\cC^{\td \leq d}}

\DeclareMathOperator{\OPT}{OPT}
\DeclareMathOperator{\td}{td}
\newcommand{\modto}[1]{\text{mod to }#1}
\newcommand{\modC}{\ensuremath{\modto{\cC}}}

\newcommand{\Nout}{N^{\texttt{out}}}

\newcommand{\conf}[2]{\textsc{conf}_{#1}(#2)}
\newcommand{\confact}[2]{\textsc{active}_{#1}(#2)}

\newcommand{\probname}[1]{\lowercase{\textsc{#1}}}
\newcommand{\myproblem}[1]{\probname{#1}\xspace}
\newcommand{\parameterizedproblem}[2]{\textup{\probname{#1}[#2]}\xspace}
\newcommand{\MODC}{\myproblem{\modC}}

\newcommand{\dominatingset}{\myproblem{Dominating Set}}
\newcommand{\dominatingsetvc}{\parameterizedproblem{Dominating Set}{vc}}

\newcommand{\feedbackvertexset}{\myproblem{Feedback Vertex Set}}
\newcommand{\feedbackvertexsetfvs}{\parameterizedproblem{Feedback Vertex Set}{fvs}}

\newcommand{\vertexcover}{\myproblem{Vertex Cover}}
\newcommand{\VC}{\ensuremath{\text{VC}}}
\newcommand{\vertexcovervc}{\parameterizedproblem{Vertex Cover}{vc}}
\newcommand{\vertexcovertdd}{\parameterizedproblem{Vertex Cover}{mod to $\td\leq d$}}
\newcommand{\vertexcoverfvs}{\parameterizedproblem{Vertex Cover}{fvs}}

\newcommand{\clusterediting}{\myproblem{Cluster Editing}}
\newcommand{\clustereditingmodcvd}{\parameterizedproblem{Cluster Editing}{cvd}}
\newcommand{\clustereditingparamce}{\parameterizedproblem{Cluster Editing}{ce}}
\newcommand{\CE}{\ensuremath{\text{CE}}}

\newcommand{\treedeletion}{\myproblem{Tree Deletion Set}}
\newcommand{\treedeletionvc}{\parameterizedproblem{Tree Deletion Set}{vc}}
\newcommand{\treedeletiontds}{\parameterizedproblem{Tree Deletion Set}{tds}}
\newcommand{\TDS}{\ensuremath{\text{TDS}}}

\newcommand{\longcycle}{\myproblem{Long Cycle}}
\newcommand{\longcyclevc}{\parameterizedproblem{Long Cycle}{vc}}
\newcommand{\longcycledeg}{\parameterizedproblem{Long Cycle}{$\#v, \text{deg}(v) \neq 2$}}
\newcommand{\LC}{\ensuremath{\text{LC}}}
\newcommand{\longpath}{\myproblem{Long Path}}
\newcommand{\longpathvc}{\parameterizedproblem{Long Path}{vc}}
\newcommand{\longpathdeg}{\parameterizedproblem{Long Path}{$\#v, \text{deg}(v) \neq 2$}}
\newcommand{\LP}{\ensuremath{\text{LP}}}

\newcommand{\hamcyclevc}{\parameterizedproblem{Hamiltonian Cycle}{vc}}
\newcommand{\hamcycledeg}{\parameterizedproblem{Hamiltonian Cycle}{$\#v, \text{deg}(v) \neq 2$}}

\newcommand{\hampathvc}{\parameterizedproblem{Hamiltonian Path}{vc}}
\newcommand{\hampathdeg}{\parameterizedproblem{Hamiltonian Path}{$\#v, \text{deg}(v) \neq 2$}}
\newcommand{\hamcyclepathdeg}{\parameterizedproblem{Hamiltonian Cycle/Path}{$\#v, \text{deg}(v) \neq 2$}}

\newcommand{\maximumcut}{\myproblem{Maximum Cut}}
\newcommand{\maximumcutmodvc}{\parameterizedproblem{Maximum Cut}{vc}}
\newcommand{\MC}{\ensuremath{\text{MC}}}

\newcommand{\steinertree}{\myproblem{Steiner Tree}}

\title{Boundaried Kernelization}
\author{
	Leonid Antipov\\Humboldt-Universität zu Berlin, Berlin, Germany\\leonid.antipov@hu-berlin.de 
	\and
	Stefan Kratsch\\Humboldt-Universität zu Berlin, Berlin, Germany\\kratsch@informatik.hu-berlin.de
}

\begin{document}
\maketitle
		
\begin{abstract}
	The notion of a (polynomial) kernelization from parameterized complexity is a well-studied model for efficient preprocessing for hard computational problems. By now, it is quite well understood which parameterized problems do or (conditionally) do not admit a polynomial kernelization. Unfortunately, polynomial kernelizations seem to require strong restrictions on the \emph{global structure} of inputs.
	
	To avoid this restriction, we propose a model for efficient \emph{local preprocessing} that is aimed at local structure in inputs. Our notion, dubbed \emph{boundaried kernelization}, is inspired by protrusions and protrusion replacement, which are tools in meta-kernelization [Bodlaender et al.\ J'ACM 2016]. Unlike previous work, we study the preprocessing of suitable boundaried graphs in their own right, in significantly more general settings, and aiming for polynomial rather than exponential bounds. We establish polynomial boundaried kernelizations for a number of problems, while \emph{unconditionally} ruling out such results for others. We also show that boundaried kernelization can be a tool for regular kernelization by using it to obtain an improved kernelization for \vertexcover parameterized by the vertex-deletion distance to a graph of bounded treedepth.
\end{abstract}


\section{Introduction}
We introduce \emph{boundaried kernelization} as a model for provably efficient \emph{local preprocessing}, i.e., preprocessing that leverages local structure without having to explicitly consider the entire instance. This extends the notion of \emph{kernelization} from \emph{parameterized complexity} and is inspired by \emph{protrusions} and \emph{protrusion replacement}, which have been used, in particular, in meta theorems for kernelization (e.g.,~\cite{DBLP:journals/jacm/BodlaenderFLPST16}). Unlike protrusions, which are subgraphs of constant treewidth that have a constant-size interface/boundary to the rest of the graph, we study different kinds of internal structure and we allow for larger boundaries. Further, we usually aim to reduce the size to at most polynomial in boundary size and optional further parameters, while for protrusion replacement it is sufficient to reduce to a size with any (usually at least exponential) dependence on treewidth and boundary size. Let us first give some context.

Parameterized complexity studies \emph{parameterized problems}, which are classical problems augmented with one or more parameters such as solution size, dimension, or treewidth; often, parameters quantify structural properties. A central goal is to design algorithms that are fast when the chosen parameter value(s) are low, showing that the corresponding property or structure is algorithmically beneficial (called \emph{fixed-parameter tractable} algorithms). This parameterized setting also allows for a robust notion of efficient preprocessing: A \emph{kernelization} for a parameterized problem is an efficient algorithm that expects an instance $(x,k)$, with parameter $k$, and returns an equivalent instance $(x',k')$ of size upper bounded by some computable function $f(k)$. In particular, one is interested what parameterized problems admit \emph{polynomial kernelizations} where the function $f$ is polynomially bounded. Kernelization has been intensely studied (see, e.g., recent works~\cite{DBLP:conf/soda/BessyBTW23,DBLP:conf/iwpec/BhyravarapuJ0S23,DBLP:conf/icalp/BougeretJS24,DBLP:conf/iwpec/Dumas023,DBLP:conf/esa/FominG00Z23,DBLP:conf/esa/FominLL0TZ23,DBLP:conf/iwpec/JansenS23,DBLP:conf/tamc/KammerS24,DBLP:conf/iwpec/KratschK23,DBLP:conf/innovations/LokshtanovM0Z24,DBLP:conf/fct/SamBGB23,DBLP:conf/iconip/XuDL23} or the dedicated book~\cite{kernelization_book_FominLSZ19}) and, by now, it is quite well understood which problems do or (conditionally) do not admit polynomial kernelizations. 
Moreover, there are both meta results (e.g.,~\cite{DBLP:journals/jacm/BodlaenderFLPST16,DBLP:journals/siamcomp/FominLST20,DBLP:journals/jcss/GanianSS16,DBLP:conf/birthday/Thilikos20}), giving upper bounds for whole classes of problems, as well as dichotomy results (e.g.,~\cite{DBLP:conf/icalp/KratschW10,DBLP:journals/toct/KratschMW16,DBLP:journals/jcss/DonkersJ21,DBLP:conf/icalp/BougeretJS24}), classifying all problems in some class into admitting or (conditionally) not admitting a polynomial kernelization. 

We see \emph{two main motivations} for studying local preprocessing: \emph{First}, kernelizations usually rely on \emph{reduction rules}, which usually apply to local configurations in the input. For example, this includes simple rules like \emph{folding of degree-two vertices} for \vertexcover~\cite{DBLP:journals/jal/ChenKJ01} and seeking \emph{flowers}, i.e., cycles that pairwise intersect in the same vertex, for \feedbackvertexset~\cite{DBLP:journals/mst/BodlaenderD10,DBLP:journals/talg/Thomasse10}. In other words, depending on the kind of reduction rule, the scope of the rule may vary, but beyond the required scope it usually does not matter what else is present in the instance.\footnote{A notably different case: Kernelizations obtained via cut-covering sets and other matroid-based techniques~\cite{DBLP:journals/jacm/KratschW20} use the entire instance to discover safe reduction steps.} Is there a sensible way of capturing local structure and giving local guarantees for the effectiveness of preprocessing?
\emph{Second}, most polynomial kernelizations in the literature are for parameterization by solution size or by the distance to some tractable special case for the problem. E.g., we are given a graph $G$ and a modulator $X\subseteq V(G)$ such that $G-X$ belongs to a graph class where the target problem can be solved efficiently, with $|X|$ being the parameter. Intuitively, unless $\P=\NP$, efficient preprocessing can at best deal with tractable parts so the parameter must directly or indirectly allow to (at least) bound the size of the remainder of the instance, which may be structured arbitrarily. Together, this yields a fairly strict requirement for kernelization to be useful: Almost the entire instance must match some tractable special case for the problem in question. What if only parts of an instance fulfill such a restriction? Do we still get nontrivial guarantees for efficient preprocessing?

\paragraph{Our work.}
Inspired by protrusions and protrusion replacement from meta kernelization~\cite{DBLP:journals/jacm/BodlaenderFLPST16}, we use \emph{boundaried graphs} to define \emph{boundaried kernelization} (see Section~\ref{full:section:boundariedkernelization} for formal definitions): A boundaried graph $G_B$ is a graph $G$ with the additional specification of a set $B\subseteq V(G)$ of boundary vertices. 
Informally, a boundaried kernelization for, e.g., the problem \vertexcoverfvs, i.e., \vertexcover parameterized by the size of a given feedback vertex set, is a polynomial-time algorithm that expects as input a boundaried graph $G_B$ with forest-modulator $X$ of size $k$, and returns a boundaried graph $G'_B$ of size bounded by a function of parameter $k$ and boundary size $|B|$ such that attaching the same graph $H_B$ to either $G_B$ or $G'_B$ (by identification of boundary vertices, also called \emph{gluing})  will yield graphs that are equivalent with respect to \vertexcover. In this way, local preprocessing for \vertexcover with respect to local structure as quantified by modulator $X$ corresponds to a boundaried kernelization for \vertexcoverfvs. Note that for $B=\emptyset$ and gluing the empty graph to $G_B$ this includes the requirement that $G$ and $G'$ are equivalent for \vertexcover (and $G'$ is of size bounded by a function of $k$), which essentially constitutes a regular kernelization for \vertexcoverfvs (cf.\ \Cref{full:lem:PBK-PK}). Accordingly, we mainly work with problems that admit a polynomial kernelization.

We establish polynomial size boundaried kernelizations for several problems that are known to admit regular polynomial kernelizations. This confirms the intuition that the local nature of reduction rules often allows us to work on the well-structured part only and ignore the rest of the instance.

\begin{theorem}\label{full:thm:positive}
	The following parameterized problems admit a polynomial boundaried kernelization: \vertexcovervc, \vertexcoverfvs, \feedbackvertexsetfvs, \longcyclevc, \longpathvc, and \hamcyclevc, \hampathvc as well as \hamcyclepathdeg.
\end{theorem}

We complement this by proving \emph{unconditionally} that some problems do not admit polynomial boundaried kernelizations despite admitting a regular polynomial kernelization. To round out the picture we also prove unconditionally that \dominatingsetvc, i.e., \dominatingset parameterized by the size of a given vertex cover, admits no polynomial boundaried kernelization, while it was already known to not admit a polynomial kernelization unless \containment. 

\begin{theorem}\label{full:thm:negative}
	The following parameterized problems do not admit a polynomial boundaried kernelization: \clustereditingmodcvd, \clustereditingparamce, \maximumcutmodvc, \treedeletionvc, \treedeletiontds, \longcycledeg, \longpathdeg, \dominatingsetvc.
\end{theorem}

Finally, we show that, like protrusion replacement, boundaried kernelization can itself be useful for obtaining regular kernelization results. Concretely, we improve the size of a kernelization for \vertexcovertdd, i.e., \vertexcover parameterized by the size of a given deletion set/modulator to treedepth at most $d$, from $\Oh(\cramped{k^{2^{\Theta(d^2)}}})$ \cite{DBLP:journals/algorithmica/BougeretS19,DBLP:journals/siamdm/HolsKP22} to $\Oh(\cramped{k^{2^{d-1}}})$ vertices, which is a significant improvement and comes much closer to the known lower bound of $\Oh(\cramped{k^{2^{d-2}+1}})$ vertices \cite{DBLP:journals/siamdm/HolsKP22}.

\begin{theorem}\label{full:thm:vc_tdd}
	For any constant $d \in \N_{\geq 1}$, the parameterized problem \vertexcovertdd admits a polynomial kernelization with $\Oh(k^{2^{d-1}})$ vertices.
\end{theorem}

\paragraph{Related work.}
Arnborg et al.~\cite{DBLP:journals/jacm/ArnborgCPS93} used finite index of gluing-equivalence (defined in \Cref{full:section:boundariedkernelization}) and protrusion replacement for the construction of algorithms that solve the membership problem for MSO-definable graph classes with bounded treewidth. Likewise, Fellows and Langston~\cite{DBLP:conf/focs/FellowsL89} showed that under condition of finite index of gluing-equivalence, the obstruction set of a minor-closed graph class can be computed as soon as a treewidth bound for the former is known. Bodlaender and van Antwerpen-de Fluiter~\cite{DBLP:journals/iandc/BodlaenderF01} extended the use to several decision, construction, and optimization problems on graphs of bounded treewidth.

As a tool for obtaining linear kernelizations for \myproblem{(Connected) Dominating Set} on graphs with a fixed excluded topological minor, Fomin et al.~\cite{DBLP:journals/talg/FominLST18} introduce a notion dubbed ``generalized protrusion:'' This refers to a subgraph with constant-size boundary but, instead of having constant treewidth, it contains at most a constant number of vertices from the sought solution (presently, from the sought minimum (connected) dominating set).
In a work on kernelization for \myproblem{Integer Linear Program Feasibility}, Jansen and Kratsch~\cite{DBLP:conf/esa/JansenK15} consider shrinking protrusion-like subsystems of ILPs with a constant-size boundary that are either totally unimodular or have a Gaifman graph of bounded treewidth.

Other prior work has considered ``dynamic sketching''~\cite{DBLP:conf/fsttcs/AssadiKLT15}, respectively ``preprocessing under uncertainty''~\cite{DBLP:conf/stacs/FafianieKQ16,DBLP:conf/mfcs/FafianieHKQ16}, mostly for tractable problems, where parts of the input are dynamic or unknown and the preprocessing needs to be consistent with any possible instantiation of the dynamic/unknown part, e.g., adjacency of vertices in dynamic part may change. Here the dynamic/unknown part is similar in spirit to the concept of a boundary, though the effects that either can have on the overall solution are very different.

\paragraph{Organization.}
\Cref{full:section:preliminaries} contains preliminaries regarding (graph) problems. In \Cref{full:section:boundariedkernelization} we define boundaried graphs, introduce boundaried kernelization and obtain some first results.
Several polynomial boundaried kernelization results are given in \Cref{full:section:upperbound}, while in \Cref{full:section:lowerbound} we unconditionally rule out existence of (polynomial) boundaried kernelization for several problems. In \Cref{full:section:VC_constTD} we apply the notion of boundaried kernelization in order to show an improved polynomial kernelization for \vertexcovertdd. We conclude in Section~\ref{full:section:conclusion}.


\section{Preliminaries}\label{full:section:preliminaries}    
\paragraph{Decision and optimization problems.}
A \textit{decision problem} is a set $\Pi \subseteq \Sigma^*$. We call any $x \in \Sigma^*$ an \textit{instance} of $\Pi$.
An \textit{optimization problem} is a function $\Pi \colon \Sigma^*  \times \Sigma^* \to \N \cup \{\pm \infty\}$. Again, an instance of $\Pi$ is defined as any $x \in \Sigma^*$. For some instance $x$ and some $s \in \Sigma^*$, if $\Pi(x, s) \neq \pm\infty$, i.e., $\Pi(x, s) \in \N$, then $s$ is called a \textit{(feasible) solution} for $\Pi$ on $x$, and $\Pi(x, s)$ is called the \textit{value} of $s$.
Optimization problems are divided into \textit{minimization} and \textit{maximization} problems. For a minimization problem we define the \textit{optimum value} for $\Pi$ on instance $x$, denoted $\OPT_{\Pi}(x)$, or simply $\OPT(x)$ if $\Pi$ is clear from context, as the lowest value of a feasible solution $s$ on $x$. Conversely, for a maximization problem we define $\OPT_{\Pi}(x)$ as the highest value of a feasible solution on $x$. If there is no feasible solution for $\Pi$ on $x$, then $\OPT_\Pi(x) = +\infty$ for a mininimization problem, resp. $\OPT_\Pi(x) = - \infty$ for a maximization problem. A solution with optimum value is called an \textit{optimum solution}.
For optimization problems used throughout this work, the \emph{value} $\Pi(x, s)$ of solution $s$ on instance $x$ is simply the \emph{size} of $s$ (denoted $|s|$). We remark, however, that in general this does not need to be the case.
We call an optimization problem $\Pi$ an \textit{$\NP$-optimization problem}, if there exists some polynomial function $p$, such that for every instance $x \in \Sigma^*$ the following are upper bounded by $p(|x|)$: (i) the size of any feasible solution $s \in \Sigma^*$ on $x$; (ii) the time needed to compute $\Pi(x, s)$ for any feasible $s$; (iii) the time needed to check if $s \in \Sigma^*$ with $|s| \leq p(|x|)$ is feasible. Note that by (ii), for such $\Pi$ the value of a feasible solution $s$ for $x$ can be encoded in length at most $p(|x|)$.
For any optimization problem $\Pi$ we define the decision problem $\Pi_d$ such that for any instance $(x, k) \in \Sigma^* \times \N$ it holds that $(x, k) \in \Pi_d$ if and only if there exists a feasible solution with value at most $k$ for minimization problems, resp., at least $k$ for maximization problems. 

\paragraph{Parameterized complexity.}
A \textit{parameterized problem} is a set $Q \subseteq \Sigma^* \times \N$. For a tuple $(x, k) \in \Sigma^* \times \N$, $k$ is called the \textit{parameter}.
For an optimization problem $\Pi$, its \emph{standard parameterization} is the decision variant $\Pi_d$ with sought solution value $k$ as the parameter, denoted by $\Pi[k]$. In contrast to this, for $\Pi$ being some decision/optimization problem and $\rho$ some minimization problem, we define the \textit{structurally parameterized problem} $\Pi[\rho]$, for which it holds that $(x, s, k) \in \Pi[\rho]$ if and only if $s$ is a feasible solution of value at most $k$ for $\rho$ on $x$, and (i) for $\Pi$ being a decision problem, $x \in \Pi$; (ii) for $\Pi$ being an optimization problem, $x \in \Pi_d$, i.e., $x = (x', \ell)$ where $\ell$ is the sought solution value for $\Pi$ on $x'$.

Let $Q$ be a parameterized problem and let $f$ be a computable function. A \textit{kernelization} for $Q$ of \textit{size} $f$ is a polynomial-time algorithm which, given as input $(x, k)$, outputs $(x', k')$ such that $|x|, k' \leq f(k)$ and $(x, k) \in Q \Leftrightarrow (x', k') \in Q$. If $f$ is a polynomial, the algorithm is called a \textit{polynomial kernelization}. We will use the term \textit{kernel} interchangeably with the term \textit{kernelization}.

\paragraph{Graphs and graph problems.}
We use standard graph theoretic notation mostly following Diestel \cite{books/Diestel17}. We denote the classes of independent and complete graphs, trees and forests by $\Cindependent, \Ccomplete, \Ctree$ and $\Cforest$, respectively. Any graph class closed under taking subgraphs is said to be \textit{hereditary}. A \textit{graph problem} is a problem for which every instance $x = (G, \hat{x})$ consists of a graph $G$ and the remaining string $\hat{x}$. We say that $\Pi$ is a \textit{pure graph problem} if the remaining string $\hat{x}$ is empty, i.e., every instance is simply a graph $G$, and $\Pi$ is isomorphism-invariant. For some graph class $\mathcal{C}$ and graph $G$, a $\mathcal{C}$-modulator (or vertex deletion set to $\cC$) for $G$ is a vertex set $X \subseteq V(G)$ such that $G - X \in \mathcal{C}$.  We denote the problem of finding a minimum size modulator to $\mathcal{C}$ by $\MODC$. All other problems discussed in this work are defined in \Cref{app:problems}.

The \textit{treedepth} of a graph $G$, denoted $\td(G)$, is the minimum height of a rooted forest $F$, called a \textit{treedepth decomposition} of $G$, such that if $\{u, v\}$ is an edge in $G$, then $u$ is an ancestor of $v$ or vice versa. There is also a recursive definition for the treedepth of $G$, which basically constructs the treedepth decomposition $F$ of $G$: (i) if $G$ is an independent set, then $\td(G) = 1$; (ii) else if $G$ is disconnected, then $\td(G) = \max\{\td(C) \mid C \text{ is a component of } G\}$; (iii) else $\td(G) = \min\{\td(G - v) + 1 \mid v \in V(G)\}$. Similar to our other graph class notation, we denote by $\Ctdd$ the class of graphs whose treedepth is at most $d$. 


\section{Boundaried kernelization}\label{full:section:boundariedkernelization}
\subsection{Boundaried graphs and related terms}
The notion of boundaried graphs and many of the related definitions were already given, e.g., by Bodlaender and van Antwerpen-de Fluiter \cite{DBLP:journals/iandc/BodlaenderF01}, and extensively used for the notion of protrusions and the related meta kernelization \cite{DBLP:journals/jacm/BodlaenderFLPST16}. 
We state our definitions on this topic and follow them by two technical but simple lemmas.

A \textit{($B$-)boundaried graph} consists of a graph $G$ and a \textit{boundary} $B \subseteq V(G)$, and is denoted as $G_B$. Two boundaried graphs $G_B, G'_B$ are \textit{isomorphic}, if there exists an isomorphism from $G$ to $G'$ that identifies vertices in $B$. 
A boundaried graph $G'_B$ is said to be a subgraph of $G_B$, if $G'$ is a subgraph of $G$.

For graph class $\cC$ we denote by $\mathcal{C}_B$ the class of all $\mathcal{C}$-modulated $B$-boundaried graphs, i.e., all boundaried graphs $G_B$ for which it holds that $G-B \in \mathcal{C}$.

Given two boundaried graphs $G_B$ and $H_C$, the \textit{gluing} operation on those boundaried graphs, denoted as $G_B \oplus H_C$, results in a new (simple) graph, which is the disjoint union of $G$ and $H$ while identifying vertices from $B \cap C$. In particular, if an edge $e$ between vertices in $B \cap C$ is contained in both $G$ and $H$, the resulting graph contains $e$ only once.
For convenience, we will tacitly assume that $V(G) \cap V(H) \subseteq B \cap C$, unless mentioned otherwise.
A few examples of graph gluing are given in \Cref{full:fig:ds_ce}.
We remark that for boundaried graphs $F_A, G_B, H_C$ clearly $(F_A \oplus G_B)_{A \cup B} \oplus H_C$ is isomorphic to $F_A \oplus (G_B \oplus H_C)_{B \cup C}$ and that $F_A \oplus G_B$ is isomorphic to $G_B \oplus F_A$, i.e., in this sense graph gluing is associative and commutative.  It also preserves isomorphism, i.e., if $G_B$ is isomorphic to $G'_B$ and $H_C$ to $H'_C$, then also $G_B \oplus H_C$ is isomorphic to $G'_B \oplus H'_C$.

Let $\Pi$ be a pure graph problem, let $G_B$ and $G'_B$ be boundaried graphs. We say that $G_B$ and $G'_B$ are \textit{gluing equivalent with respect  to $\Pi$ and $B$}, denoted $G_B \equiv_{\Pi, B} G'_B$, if: (i)  for $\Pi$ being a decision problem, it holds for every $H_B$, that $G_B \oplus H_B \in \Pi \Leftrightarrow G'_B \oplus H_B \in \Pi$; (ii) for $\Pi$ being an optimization problem, there exists some constant $\Delta \in \Z$ such that for every $H_B$ it holds that $\OPT_\Pi(G_B \oplus H_B) = \OPT_\Pi(G'_B \oplus H_B) + \Delta$. In the latter case we also say that $G_B$ and $G'_B$ are gluing equivalent with \textit{offset} $\pm \Delta$, tacitly having fixed one direction.
For some graph class $\mathcal{C}$ the equivalence relation $\equiv_{\Pi, B}^\mathcal{C}$ (resp. $\equiv_{\Pi,B}^{\mathcal{C}_B}$) is the restriction of $\equiv_{\Pi, B}$ to $\mathcal{C}$ (resp. $\mathcal{C}_B$). If $\Pi$, $B$, and $\cC$ or $\cC_B$ are clear from context, we write $\equiv$.

Some pure graph decision problem $\Pi$ has \textit{finite index}, if there exists some function $f \colon \N \to \N$, such that for each boundary $B$ the number of equivalence classes of $\equiv_{\Pi, B}$ is at most $f(|B|)$. In the case of $\Pi$ being an optimization problem, $\Pi$ is said to have \textit{finite integer index}. For graph class $\mathcal{C}$ we say that $\Pi$ has \textit{finite (integer) index on $\mathcal{C}$} (resp.\ on $\mathcal{C}_B$), if there exists some function $f \colon \N \to \N$ such that for each boundary $B$ the number of equivalence classes of $\equiv_{\Pi, B}^{\mathcal{C}}$ (resp.\ $\equiv_{\Pi, B}^{\mathcal{C}_B}$) is at most $f(|B|)$.
Among problems with finite (integer) index, we highlight those having \textit{single-exponential (integer) index}, which in addition requires some constant $c$ and polynomially bounded function $f$ such that $\equiv_{\Pi, B}$ (resp.\ $\equiv_{\Pi, B}^{\mathcal{C}}$ and $\equiv_{\Pi, B}^{\mathcal{C}_B}$) has $\mathcal{O}(c^{f(|B|)})$ equivalence classes.

The following lemma allows us to use gluing-equivalence for an artificially increased boundary in order to obtain automatically also gluing-equivalence for the actually needed boundary. After that, we state \Cref{full:lem:glue_equiv_get_equiv} which will be used in \Cref{full:section:VC_constTD}. Note that problem $\Pi$ can be either a decision or an optimization problem. For the proof of \Cref{full:lem:BD} we write the proof for both cases. However, since there is not much additional argumentation needed for the case of a decision problem, in subsequent proofs we will assume that $\Pi$ is an optimization problem.

\begin{lemma}\label{full:lem:BD}
	Let $\Pi$ be a pure graph problem, $G$ and $G'$ graphs, and $B, C$ vertex subsets of both $V(G)$ and $V(G')$ such that $B \subseteq C$.
	Then $G_C \equiv_{\Pi, C} G'_C$ implies $G_B \equiv_{\Pi, B} G'_B$, with the same offset $\Delta$ for these two equivalences, if $\Pi$ is an optimization problem.
\end{lemma}
\begin{proof}
	Assume first that $\Pi$ is an optimization problem and that the stated condition holds, i.e., for some constant $\Delta$ and every $H_C$ it holds that $\OPT_\Pi(G_C \oplus H_C) = \OPT_\Pi(G'_C \oplus H_C) + \Delta$. Now to show that for any $H_B$ it also holds that $\OPT_\Pi(G_B \oplus H_B) = \OPT_\Pi(G'_B \oplus H_B) + \Delta$, fix an arbitrary $B$-boundaried graph $H_B$. Since $\Pi$ is isomorphism invariant, we assume without loss of generality, that the vertices in common between $H$ and $G$, resp., $H$ and $G'$, are all  contained in $B$. As a result, also no vertices in $C \setminus B$ are in common between $H$ and $G, G'$. From $H$ we construct the graph $H'$ by adding the vertices in $C \setminus B$ as an independent set, which leads to the fact that $G_B \oplus H_B = G_C \oplus H'_C$ and $G'_B \oplus H_B = G'_C \oplus H'_C$. This way we obtain the needed equivalence $\OPT_\Pi(G_B \oplus H_B) = \OPT_\Pi(G_C \oplus H'_C) = \OPT_\Pi(G'_C \oplus H'_C) + \Delta = \OPT_\Pi(G'_B \oplus H_B) + \Delta$.\\
	Now assume that $\Pi$ is a decision problem. Again, for any $H_B$ it holds that $G_B \oplus H_B = G_C \oplus H'_C$ and $G'_B \oplus H_B = G'_C \oplus H'_C$, which makes $G_B \oplus H_B \in \Pi \Leftrightarrow G'_B \oplus H_B \in \Pi$ follow from the fact that $G_C \oplus H'_C \in \Pi \Leftrightarrow G'_C \oplus H'_C \in \Pi$.
\end{proof}

\begin{lemma}\label{full:lem:glue_equiv_get_equiv}
	Let $\Pi$ be a pure graph problem, $B$ some vertex set, and $G^1$, $G^2$, $\hat{G}^1$, $\hat{G}^2$ graphs, such that $G^1_B \equiv_{\Pi, B} \hat{G}^1_B$ and $G^2_B \equiv_{\Pi, B} \hat{G}^2_B$. 
	Then it holds that $(G^1_B \oplus G^2_B)_B \equiv_{\Pi, B} (\hat{G}^1_B \oplus \hat{G}^2_B)_B$. \\
	If $\Pi$ is an optimization problem and $\Delta_1, \Delta_2$ are the offset values for the equivalences between $G^1_B$ and $\hat{G}^1_B$, resp. $G^2_B$ and $\hat{G}^2_B$, when glued to any $H_B$, then the offset for the equivalence between $(G^1_B \oplus G^2_B)_B$ and $(\hat{G}^1_B \oplus \hat{G}^2_B)_B$ equals $\Delta_1 + \Delta_2$.
\end{lemma}
\begin{proof}
	Without loss of generality assume that $\Pi$ is an optimization problem. Let $H_B$ be any $B$-boundaried graph.
	From $G^1_B \equiv_{\Pi, B} \hat{G}^1_B$ follows existence of some $\Delta_1$ such that $\OPT(G^1_B \oplus (G^2_B \oplus H_B)_B) = \OPT(\hat{G}^1_B \oplus (G^2_B \oplus H_B)_B)  + \Delta_1$. At the same time $G^2_B \equiv_{\Pi, B} \hat{G}^2_B$ implies some $\Delta_2$ such that $\OPT(G^2_B \oplus (\hat{G}^1_B \oplus H_B)_B) = \OPT(\hat{G}^2_B \oplus (\hat{G}^1_B \oplus H_B)_B)  + \Delta_2$. 
	Together with commutativity  and associativity of graph gluing, this yields $\OPT((G^1_B \oplus G^2_B)_B \oplus H_B)) = \OPT((\hat{G}^1_B \oplus G^2_B)_B \oplus H_B) + \Delta_1 = \OPT((\hat{G}^1_B \oplus \hat{G}^2_B)_B \oplus H_B) + \Delta_1 + \Delta_2$.
\end{proof}

\subsection{Boundaried kernelization}
We will now formally introduce our notion of boundaried kernelization. Presently, we do this for the restricted setting of pure graph problems $\Pi$ parameterized by some pure graph minimization problem $\rho$. This enables us to prove a few general results on boundaried kernelizations and captures most structural parameterizations while also keeping the required notation somewhat concise. 
Possible extensions of the notion are briefly discussed in the conclusion.
We remark that while we work with pure graph problems, the boundaried kernelization gets as input a \emph{boundaried} graph $G_B$. This is consistent since the result of gluing $G_B$ with any other boundaried graph $H_B$ results in the \emph{simple} graph $Q = G_B \oplus H_B$, which is then the actual input to the pure graph problem, for which the preprocessing was done on the known local part $G_B$.

Unlike in regular kernelization, where, for $\Pi$ being an optimization problem, a sought solution size $k$ is given and can be changed by the kernel (adjusting to situations like safeness of including some vertex in the solution), this is not the case for boundaried kernelization, where only local structure of $G_B$ is used, independently of what graph $H_B$ could be glued to $G_B$ and what the global solution size could be. For this reason, the boundaried kernelization needs to output its influence on the solution size in form of the solution size offset $\Delta$, which does not need to be bounded. Note that this $\Delta$ is already part of the definition of gluing equivalence.

\begin{definition}[Boundaried kernelization]\label{full:def:bk}
	Let $\Pi$ be a pure graph problem, $\rho$ a pure graph minimization problem, and $f$ a computable function. A \textit{boundaried kernelization} of size $f$ for $\Pi[\rho]$ is a polynomial-time algorithm that, given a boundaried graph $G_B$ and $\rho$-solution $s$, outputs a $\equiv_{\Pi, B}$-equivalent boundaried graph $G'_B$ with $\rho$-solution $s'$, such that $|G'|, |s'|, \rho(G', s') \leq f(|B| + \rho(G, s))$. Additionally, if $\Pi$ is an optimization problem, the algorithm also needs to output an integer $\Delta$, such that for all $H_B$ it holds that $\OPT_\Pi(G_B \oplus H_B) = \OPT_\Pi(G'_B \oplus H_B) + \Delta$. If $f$ is upper bounded by some polynomial function, the boundaried kernelization has \textit{polynomial size} and is called a \textit{polynomial boundaried kernelization}.
\end{definition}

Later we will show polynomial boundaried kernelizations for several problems parameterized by the size of a given modulator $|X|$ to some hereditary graph class $\mathcal{C}$. As the size of the output should be bounded polynomially in $|B| + |X|$, it suffices to bound/reduce the size of $R := V(G) \setminus (B \cup X)$. In order to simplify this, we show that it suffices to deal with the special case where $B$ itself is a modulator to $\mathcal{C}$, i.e., where $G-B\in\mathcal{C}$.

\begin{lemma}\label{full:lem:Bismod}
	Let $\Pi$ be a pure graph problem, $\cC$ be a hereditary graph class and $f$ a computable function. If there exists a boundaried kernelization of size $f(|B|)$ for $\Pi[\modC]$ with the additional restriction that $B$ coincides with the given and output solutions to \MODC, then there also exists a boundaried kernelization of size $f$ for $\Pi[\modC]$ without that restriction.
\end{lemma}
\begin{proof}
	Without loss of generality assume that $\Pi$ is an optimization problem. Let $G_B$ be the boundaried graph and $X$ the $\cC$-modulator, which are given to the unrestricted boundaried kernelization. Since $\cC$ is hereditary per requirement, it also holds that $B \cup X$ is a $\cC$-modulator. Thus we can apply the restricted boundaried kernelization on boundaried graph $G_{B \cup X}$ and $\cC$-modulator $B \cup X$, which outputs a boundaried graph $G'_{B \cup X}$, the $\cC$-modulator $B \cup X$, and an integer $\Delta$, such that $|G'|, |B \cup X| \leq f(|B \cup X|)$ and for all $H_{B \cup X}$ it holds that $\OPT_{\Pi}(G_{B \cup X} \oplus H_{B \cup X}) = \OPT_\Pi(G'_{B \cup X} \oplus H_{B \cup X}) + \Delta$. By \Cref{full:lem:BD} it becomes clear that for the same $\Delta$ it also holds for every $H_B$ that $\OPT_{\Pi}(G_B \oplus H_B) = \OPT_\Pi(G'_B \oplus H_B) + \Delta$. Therefore it suffices for the unrestricted boundaried kernelization to output the boundaried graph $G'_B$ and its $\cC$-modulator $B \cup X$.
\end{proof}

The following lemma shows that, for parameterization by \NP-minimization problem $\rho$, a boundaried kernelization of size $f$ also implies a (regular) kernelization of size at most polynomial in $f$. In this sense, boundaried kernelizations are stronger than (regular) kernelizations, showing, in particular, that a polynomial kernelization is a prerequisite for admitting a polynomial boundaried kernelization. We note that most structural parameterizations are indeed based on \NP-minimization problems, as one wishes to be able verify the (often provided) structure efficiently. 
(Parameterization by maximization problems is often unwieldy to deal with, e.g., because we need optimum solutions to get structural implications, but optimality is usually hard to verify. Nevertheless, parameterization by max leaf number comes to mind as a useful exception, cf.~\cite{DBLP:journals/mst/FellowsLMMRS09}.)

\begin{lemma}\label{full:lem:PBK-PK}
	Let $\Pi$ and $\rho$ be pure graph problems. Additionally, let $\rho$ be an $\NP$-minimization problem, and $\Pi$ either a decision problem or an $\NP$-optimization problem.
	If $\Pi[\rho]$ admits a (polynomial) boundaried kernelization, it also admits a (polynomial) kernelization.
\end{lemma}
\begin{proof}
	First, we prove the lemma for $\Pi$ being an $\NP$-optimization problem. We are given an instance $(G, \ell, s, k)$, i.e., $s$ is a solution with value at most $k$ for $\rho$ on $G$ and the question is if $\Pi$ has a solution for $G$ with value at most $\ell$ if $\Pi$ is a minimization problem, resp. with value at least $\ell$ if $\Pi$ is a maximization problem. Let us first make sure that the encoding of $\ell$ is not overly large: By the definition of $\NP$-optimization problems, there exists some polynomial function $p$ for $\Pi$, such that any feasible solution for $\Pi$ on $G$ has size at most $p(|G|)$ and its value can be encoded with length at most $p(|G|)$. Hence, if $\ell$ is too large for such an encoding, and if $\Pi$ is a minimization problem, then there exists a feasible solution with value at most $\ell$ if and only if there exists some solution with a value that can be encoded with length $p(|G|)$, so we can work further with the instance $(G, \ell', s, k)$ instead, where $\ell'$ is the highest possible number that can be encoded with length $p(|G|)$; and if $\Pi$ is a maximization problem, then there cannot be a  feasible solution with value at least $\ell$, so we can output some fixed NO-instance of constant size. So from now it holds for our instance $(G, \ell, s, k)$ that $\ell$ can be encoded with length at most $p(|G|)$.
	
	For any boundary $B \subseteq V(G)$, gluing $G_B$ to the independent set on vertex set $B$ results in $G$ itself. In particular this also works for $B$ being empty.
	We apply the boundaried kernelization of size $f$ on $G_{\emptyset}$ and $s$, and get as output a boundaried graph $G'_{\emptyset}$, a $\rho$-solution $s'$, and an integer $\Delta$ such that for all $H_\emptyset$ it holds that $\OPT_\Pi(G_\emptyset \oplus H_\emptyset) = \OPT_\Pi(G'_\emptyset \oplus H_\emptyset) + \Delta$. By choosing $H_\emptyset$ as the empty graph, it holds in particular that $\OPT_\Pi(G) = \OPT_\Pi(G') + \Delta$. Both $G'_{\emptyset}$ and $s'$ have size and value at most $f(0 + \rho(G, s)) \leq f(k)$. If $\Delta > \ell$ and $\Pi$ is a minimization problem, we immediately recognize that there cannot be a feasible solution with value at most $\ell - \Delta<0$ for $G'$ and thus, by gluing equivalence, no feasible solution with value at most $\ell$ exists for $G$. Conversely, if $\Delta > \ell$ and $\Pi$ is a maximization problem, then there exists a feasible solution with value at least $\ell - \Delta$ if and only if there exists one with value at least $0$; so we can replace $\Delta$ by $\ell$. Similar to the first paragraph, we can now make sure that $\ell - \Delta$ can be encoded with length at most $p(|G'|) \leq p(f(k))$, and otherwise replace $\ell-\Delta$ by the highest possible number that can be encoded with length $p(|G'|)$ or output a fixed NO-instance, depending on whether $\Pi$ is a minimization or maximization problem. Altogether, the instance $(G', \ell - \Delta, s', \rho(G', s'))$ is a correct output for a (polynomial) kernelization, as it satisfies the following: (i) $(G, \ell, s, k) \in \Pi[\rho] \Leftrightarrow (G', \ell-\Delta, s', \rho(G', s')) \in \Pi[\rho]$; and (ii) $|G'|, |s'|, \rho(G', s')$ and the encoding length of $|\ell-\Delta|$ are all upper bounded by $f(k) + p(f(k))$.
	
	Now let $\Pi$ be a decision problem. We are given an instance $(G, s, k)$ this time, with $s$ a solution with value at most $k$ for $\rho$ on $G$ and the question being if $G \in \Pi$. Similar to above, we apply the boundaried kernelization of size $f$ on $G_\emptyset$ and $s$, in order to obtain a boundaried graph $G'_\emptyset$ and its $\rho$-solution $s'$, such that for all $H_\emptyset$ it holds that $G_\emptyset \oplus H_\emptyset \in \Pi \Leftrightarrow G'_\emptyset \oplus H_\emptyset \in \Pi$. Again, by choosing $H_\emptyset$ as the empty graph, it holds in particular that $G \in \Pi \Leftrightarrow G' \in \Pi$. Both $G'$ and $s'$ have size and value at most $f(0 + \rho(G, s)) \leq f(k)$. Altogether, the instance $(G', s', \rho(G', s'))$ is a correct output for a (polynomial) kernelization, as it satisfies the following: (i) $(G, s, k) \in \Pi[\rho] \Leftrightarrow (G', s', \rho(G', s')) \in \Pi[\rho]$; and (ii) $|G'|, |s'|$ and $\rho(G', s')$ are all upper bounded by $f(k)$.
\end{proof}

The next and final lemma of this section will be used to rule out (polynomial) boundaried kernelizations. In contrast to lower bounds for (regular) kernelization (cf.\ \cite{DBLP:journals/siamdm/BodlaenderJK14,DBLP:journals/jacm/DellM14}), working with the index of the corresponding equivalence classes allows us to get unconditional lower bounds. In particular, it follows that if $\Pi[\modC]$ has a polynomial boundaried kernelization, then $\equiv_{\Pi, B}^{\mathcal{C}_B}$ and $\equiv_{\Pi, B}^{\mathcal{C}}$ have single-exponential index (and, conversely, larger index rules out such a boundaried kernelization).

\begin{lemma}\label{full:lem:BK-index}
	Let $\Pi$ be a pure graph problem, $\rho$ a pure graph minimization problem and $f$ a computable function such that $\Pi[\rho]$ admits a boundaried kernelization of size $f$. Further fix some vertex set $B$ and graph class $\cC$ such that for each $G \in \cC$ with $B \subseteq V(G)$ exists some $\rho$-solution $s$ with $\rho(G, s) \leq g(|B|)$ for some function $g$.
	Then the number of equivalence classes of $\equiv^\cC_{\Pi, B}$ is upper bounded by $\Oh(2^{f(|B| + g(|B|))^2})$.
\end{lemma}
\begin{proof}
	Let $F_B, G_B$ be two boundaried graphs with $F, G \in \cC$ that are not equivalent with respect to $\equiv^\cC_{\Pi, B}$. Let $s_F, s_G$ be respective $\rho$-solutions with value at most $g(|B|)$, which exist by assumption of the lemma. Then the given boundaried kernelization outputs for $F_B, s_F$ and $G_B, s_G$ two graphs $F'_B$ and $G'_B$ which are not isomorphic, as otherwise it would hold that $F_B \equiv_{\Pi, B} F'_B \equiv_{\Pi, B} G'_B \equiv_{\Pi, B} G_B$, leading to a contradiction.\\
	The number of equivalence classes of $\equiv^{\cC}_{\Pi, B}$ is thus upper bounded by the number of non-isomorphic graphs we might get as output from the boundaried kernelization when given some $B$-boundaried graph in $\cC$ and respective solution $s$ with guaranteed value at most $g(|B|)$. Since these graphs have size (and thus number of vertices) at most $q = f(|B| + g(|B|))$, their number is upper bounded by $\sum_{n = 0}^q\sum_{m = 0}^{\binom{n}{2}} \binom{\binom{n}{2}}{m} \leq q \cdot 2^{\binom{q}{2}} \in \Oh(2^{\cramped{q^2}})$. 
\end{proof}


\section{Upper bounds}\label{full:section:upperbound}
In this section we show a selection of problems to admit polynomial boundaried kernelization. Our algorithms will be described closely to the (regular) polynomial kernels for the respective problems, which meets the expectation for many reduction rules to be applicable and helpful for boundaried kernelization due to their rather local nature. This will be formally done by showing \textit{gluing safeness} of such rules, i.e., preserving gluing-equivalence.
As given by \Cref{full:lem:Bismod} we assume to be given a boundaried graph $G_B$ such that $G-B$ is either an independent set or a forest, depending on the parameterization. We denote the vertex set of $G-B$ by $R$. 

\subsection{Vertex Cover[vc]}\label{full:section:vc}
Our result for \vertexcovervc is based on the crown reduction rule, which was first introduced by Chen et al.~\cite{DBLP:journals/jal/ChenKJ01}. A \textit{crown} in a graph $G$ is a pair of vertex sets $I, H \subseteq G$, such that: (i) $I$ is a non-empty independent set of $G$; (ii) $H$ is the neighborhood of $I$; and (iii) there exists an $H$-saturating matching between $I$ and $H$. Abu-Khzam et al.~\cite{DBLP:journals/mst/Abu-KhzamFLS07} gave on overview of the technique.

\begin{lemma}[\cite{DBLP:journals/mst/Abu-KhzamFLS07}, Theorem 3]\label{full:lem:crown_solution}
	If $G$ is a graph with crown $(I, H)$, then there is a minimum vertex cover of $G$, containing all $H$-vertices and none from $I$.
\end{lemma}

Observe that crowns in a boundaried graph $G_B$ might become invalid after gluing to some boundaried graph $F_B$, e.g., $I$-vertices contained in $B$ might become adjacent. For this reason we need to make sure that the intersection of $I$ and $B$ is empty. Furthermore we cannot simply remove vertices of $B$, as they will be re-introduced by gluing. Instead, we fix an $H$-saturating matching between $I$ and $H$, and remove all other edges incident with $H$. As $(I, H)$ remains a crown after that operation, we mark vertices in $I$ and $H$ to prevent multiple handling, therefore initialize the set of marked vertices $\tilde{I}, \tilde{H} = \emptyset$. Additionally, we remove any isolated vertices in $R$, as even after gluing, they will not be incident with any edge.

\begin{redrule}\label{full:rr:VC_iso}
	If $v \in R$ is an isolated vertex, remove it.
\end{redrule}
\begin{proof}[Proof of gluing safeness]
	Assume the condition for the rule holds, i.e., $v \in R$ is an isolated vertex. Let $G'_B = G_B - v$ be the resulting graph. For any boundaried graph $F_B$ we know that $v$ is an isolated vertex in $G_B \oplus F_B$, as it is not part of the boundary and thus does not get any new neighbors after gluing. It holds that $\OPT_{\VC}(G_B \oplus F_B) = \OPT_{\VC}(G'_B \oplus F_B)$. 
\end{proof}

\begin{redrule}\label{full:rr:VC_crown}
	If $G - (\tilde{I} \cup \tilde{H})$ contains a crown $(I, H)$ with $I \subseteq R$ and an $H$-saturating matching $M$ between $H$ and $I$, then remove all edges incident with $H$ but not contained in $M$, and add all vertices in $I, H$ to $\tilde{I}, \tilde{H}$, respectively.
\end{redrule}
\begin{proof}[Proof of gluing safeness]
	Let $(I, H)$ be a crown that satisfies the above conditions, let $G'_B$ be the resulting graph after removing all edges incident with $H$, but not contained in $M$. For any $B$-boundaried graph $F_B$ the pair $(I,H)$ remains a crown in $G_B \oplus F_B$ since no vertex in $I$ gets new neighbors after gluing, and consequently $I$ remains an independent set, $H$ its neighborhood, and $M$ an $H$-saturating matching between $H$ and $I$.
	Clearly, removing all edges incident with $H$ besides those in $M$ does not change any of those points, and $(I, H)$ is also a crown in $G'_B \oplus F_B$. Also note that, since the only difference between $G$ and $G'$ is about the edges of $H$, we have $(G_B \oplus F_B) - (I \cup H) = (G'_B \oplus F_B) - (I \cup H)$. Due to \Cref{full:lem:crown_solution}, we then get
	\begin{align*}
		\OPT_{\VC}(G_B \oplus F_B) 
		&= \OPT_{\VC}((G_B \oplus F_B) - (I \cup H)) + |H| \\
		&= \OPT_{\VC}((G'_B \oplus F_B) - (I \cup H)) + |H| \\
		&= \OPT_{\VC}(G'_B \oplus F_B) 
		\qedhere
	\end{align*}
\end{proof}

Such restrictive handling of crowns still suffices in order to bound the independent set number of $G[R]$, as will be given by \Cref{full:lem:no_crown_small_IS}. 
Next we show how to apply \Cref{full:rr:VC_iso,full:rr:VC_crown} exhaustively in polynomial time by modification of the algorithm by Iwata et al.~\cite{DBLP:conf/soda/IwataOY14} for removing the crowns in a graph. 

First, let $\overline{G}$ be a bipartite graph, which has as vertex set the disjoint union of $\overline{X} \coloneq \{x_v \mid v \in R \setminus (\tilde{I} \cup \tilde{H})\}$ and $\overline{Y} \coloneq \{y_v \mid v \in V(G) \setminus (\tilde{I} \cup \tilde{H})\}$, and the edge set is $\{\{x_u, y_v\} \mid \{u,v\} \in E(G - (\tilde{I} \cup \tilde{H})\}$. We refer to the total vertex set of $\overline{G}$ by $\overline{V}$. For vertex set $\overline{S} \subseteq \overline{V}$ we write $S_X$ for $\{v \in R \mid x_v \in \overline{S} \cap \overline{X}\}$ and $S_Y$ for $\{v \in V \mid y_v \in \overline{S} \cap \overline{Y}\}$. The other way around, for vertex set $T \subseteq V$, we define $\overline{X}_T \coloneq \{x_v \in \overline{X} \mid v \in T\}$ and $\overline{Y}_T \coloneq \{y_v \in \overline{Y} \mid v \in T\}$.
Using $\overline{G}$, we build our actual auxiliary graph $\overline{G}_M$ using a maximum matching $M$ in $\overline{G}$. The vertex set is also $\overline{V}$, but the edges are directed copies of those in $\overline{G}$, directed from $\overline{X}$ to $\overline{Y}$. Only edges belonging to $M$ are additionally directed the other way around, i.e., from $\overline{Y}$ to $\overline{X}$.

At this point, we will also be working with directed graphs, so let us denote, with $F$ being a directed graph and $A \subseteq V(F)$, by $\Nout_F(A)$ the set of vertices in $V(F) \setminus A$, to which there is an incoming edge going from some vertex in $A$.

We show that finding a crown $(I, H)$ in $G - (\tilde{I} \cup \tilde{H})$ with $I \subseteq R$ is equivalent to finding a tail strongly connected component in $\overline{G}_M$, i.e., a set $\oS$ of vertices such that $\Nout_{\oG_M}(\oS) = \emptyset$ and $\overline{G}_M[\oS]$ contains directed paths between any two vertices in any direction.

\begin{lemma}[Analogue of {\cite[Lemma 4.5]{DBLP:conf/soda/IwataOY14}}]\label{full:lem:tscc-crown}
	If $\overline{S}$ is a tail strongly connected component of $\overline{G}_M$ with $S_X \neq \emptyset$ and $S_X \cap S_Y = \emptyset$, then $(S_X, S_Y)$ is a crown in $G - (\tilde{I} \cup \tilde{H})$ with $S_X \subseteq R$.
\end{lemma}
\begin{proof}
	Since $S_X \subseteq R$ and $(S_X \cup S_Y)  \cap (\tilde{I} \cup \tilde{H}) = \emptyset$ follows by definition, we only need to show that $(S_X, S_Y)$ is a crown. Therefore, we simply go through the three conditions of the crown definition.
	\begin{enumerate}
		\item $S_X$ is a non-empty independent set.\\
		$S_X$ is non-empty by condition of the lemma, and there might be no $u,v \in S_X$ with $\{u,v\} \in E(G)$, as otherwise it would hold that $\{x_u,y_v\} \in \overline{E}$ and thus either $S_X \cap S_Y \neq \emptyset$ or $\overline{S}$ has outgoing edges and is thus not a tail strongly connected component.
		\item $S_Y = N_G(S_X)$.\\
		As $\overline{S}$ is strongly connected, it must hold in particular that $\overline{S} \cap \overline{Y} \subseteq \Nout_{\overline{G}_M}(\overline{S} \cap \overline{X})$, and as $\overline{S}$ has no outgoing edges,  $\Nout_{\overline{G}_M}(\overline{S} \cap \overline{X}) \subseteq \overline{S} \cap \overline{Y}$. Together we get $\Nout_{\overline{G}_M}(\overline{S} \cap \overline{X}) = \overline{S} \cap \overline{Y}$ and, since no vertex outside of $\tilde{I} \cup \tilde{H}$ is adjacent to $\tilde{I} \cup \tilde{H}$, also $N(S_X) = S_Y$.
		\item There is an $S_Y$-saturating matching into $S_X$.\\
		Since $\overline{S}$ is strongly connected, there is an edge from each $y_v \in \overline{S} \cap \overline{Y}$ to some $x_u \in \overline{S} \cap \overline{X}$ in $\overline{G}_M$, which means $\{y_v, x_u\}$ is contained in the matching $M$ of $\overline{G}$ which was used to generate $\overline{G}_M$. Furthermore, by construction it holds that $\overline{G}[\overline{S}]$ is a subgraph of $G[S_X \cup S_Y]$ and thus the edge set $\{\{u,v\} \mid x_u, y_v \in \overline{S} \text{ and } \{x_u, y_v\} \in M\}$ is an $S_Y$-saturating matching into $S_X$. \qedhere
	\end{enumerate}
\end{proof}

\begin{lemma}\label{full:lem:vc_matched}
	Let $(I, H)$ be a crown in $G - (\tilde{I} \cup \tilde{H})$ with $I \subseteq R$, such that  $G[I \cup H]$ is connected. If $\oY_H$ is matched by $M$ into $\oX_I$, then there is a tail strongly connected component $\oS \subseteq \oX_I \cup  \oY_H =: \oQ$ in $\oG_M$, such that $S_X \neq \emptyset$ and $S_X \cap S_Y = \emptyset$.
\end{lemma}
\begin{proof}
	It is well known, that every finite directed graph has at least one tail strongly connected component. Let $\overline{S}$ be any tail strongly connected component of $\oG_M[\oQ]$. Note that $\overline{S}$ remains a tail strongly connected component of $\oG_M$, as $\oQ$, and thus in particular also $\overline{S}$, has no outgoing edges in $\oG_M$: Every vertex in $\oY_H$ is matched to some vertex in $\oX_I$, which implies $\Nout_{\oG_M}(\oY_H) \subseteq \oX_I$, and since $(I, H)$ is a crown, it holds that $N(I) = H$, which implies $\Nout_{\oG_M}(\oX_I) = \oY_H$.
	
	$S_X \cap S_Y = \emptyset$ is easy to see, as it holds by construction, that $S_X \subseteq I$ and $S_Y \subseteq H$, while $I \cap H = \emptyset$, since $I$ is an independent set and $H = N(I)$. Let us show that $S_X \neq \emptyset$. Assume for contradiction that $\overline{S} \subseteq \overline{Y}$ and thus $S_X = \emptyset$. As seen above, every vertex $y_v  \in \oS \cap \oY \subseteq \oY_H$ is matched by $M$ to some $x_u \in \oX_I$, and by construction of $\oG_M$ exists an edge from $y_v$ to $x_u$. As $\oS$ has no outgoing edges, $x_u$ needs to be contained in $\oS$, and we get $S_X \neq \emptyset$.
\end{proof}

\begin{lemma}\label{full:lem:vc_not_matched}
	Let $(I, H)$ be a crown in $G - (\tilde{I} \cup \tilde{H})$ with $I \subseteq R$, such that $G[I \cup H]$ is connected. If $\oY_H$ is not matched by $M$ into $\oX_I$, then each vertex $y_v \in \oY_H$ that is not matched by $M$ to any $\oX_I$-vertex, is instead matched by $M$ to some vertex in $\oX \setminus \oX_I$. Furthermore, for each such $y_v \in \oY_H$ exists some $x_u \in \oX_I$ that is not matched by $M$, and such that there is an $M$-alternating path from $y_v$ to $x_u$, whose first and last edges are not contained in $M$.
\end{lemma}
\begin{proof}
	Fix any $y_v \in \oY_H$ which is not matched by $M$ to any $\oX_I$-vertex. Since $H = N(I)$ has a saturating matching into $I$, and hence $\oY_H = N_{\oG}(\oX_I)$ has one into $\oX_I$, but $M$ does not match $y_v$ to any vertex in $\oX_I$, it must hold, as $M$ is maximum, that $y_v$ is matched by $M$ to some vertex outside of $\oX_I$.
	
	Construct vertex sets $\oY_n \subseteq \oY_H$ and $\oX_n = N_{\oG}(\oY_n) \cap \oX_I$ as follows. With $\oY_1 = \{y_v\}$, repeat for $i = 1$ to $n$ such that $\oY_{n+1} = \oY_n$: let $\oX_i = N_{\oG}(\oY_i) \cap \oX_I$ and $\oY_{i+1} = \oY_i \cup N_M(\oX_i)$. By Hall's theorem, and since $\oY_H$ has a saturating matching into $\oX_I$, it holds that $|\oY_n| \leq |\oX_n|$. At the same time,  $y_v$ is not matched to any vertex in $\oX_n \subseteq \oX_I$, so there is at least one vertex $x_u \in \oX_n$ which is not matched by $M$. By construction of $\oX_n$, there is an $M$-alternating path from $y_v$ to $x_u$, with the edges in both ends not contained in $M$.
\end{proof}

\begin{lemma}[Analogue of {\cite[Lemma 4.6]{DBLP:conf/soda/IwataOY14}}]\label{full:lem:crown-tscc}
	If one cannot apply \Cref{full:rr:VC_iso} and there is a crown in $G - (\tilde{I} \cup \tilde{H})$ with $I \subseteq R$, then there exists a tail strongly connected component $\oS$ of $\oG_M$ with $S_X, S_Y \neq \emptyset$ and $S_X \cap S_Y = \emptyset$.
\end{lemma}
\begin{proof}
	We can assume that $G[I \cup H]$ is connected. Otherwise, choose any connected component $T$ of $G[I \cup  H]$, and let $I' = I \cap T$ and $H' = H \cap T$. It clearly holds that $H' = N(I')$ and that there exists an $H'$-saturating matching into $I'$. Since we cannot apply \Cref{full:rr:VC_iso}, it also holds that $H' \neq \emptyset$. Thus, $(I', H')$ is a crown in $G - (\tilde{I} \cup \tilde{H})$ with $I' \subseteq R$, such that $G[I' \cup H']$ is connected.
	
	If $\oY_H$ is matched by $M$ into $\oX_I$, then we can apply \Cref{full:lem:vc_matched} and are thus done. Hence, assume that $\oY_H$ is not matched by $M$ into $\oX_I$.
	Construct as follows a vertex set $\oQ = \oX_n \cup \oY_n$ with $n \in \N$, such that $\oQ$ induces a (weakly) connected graph in $\oG_M$ and has no outgoing edges outside of itself. Set $\oX_1 = \oX_I$ and repeat for $i = 1$ to $n$, such that $\oX_{n+1} = \oX_n$: let $\oY_i = N_{\oG}(\oX_i)$ and $\oX_{i+1} = \oX_i \cup N_M(\oY_i)$. Observe, that every $y_v \in \oY_n$ is matched by $M$ to some vertex in $\oX_n$: Assume for contradiction that some $y_v \in \oY_n$ is not matched by $M$ to some vertex in $\oX_n$. Then by construction of $\oX_n$, it must hold, that $y_v$ is not matched at all by $M$. By \Cref{full:lem:vc_not_matched} and maximality of $M$ it cannot hold that $y_v \in \oY_H$, so we can assume that $y_v \in \oY_n \setminus \oY_H$. By construction of $\oY_n$ it holds that $y_v$ has an $M$-alternating path to some $y_w \in \oY_1 = \oY_H$ which is matched outside of $\oX_I$. Again by using \Cref{full:lem:vc_not_matched}, it then holds that $y_v$ has an $M$-alternating path,  through $y_w$, to some $x_u \in \oX_I$ with both $y_v$ and $x_u$ not being matched by $M$, which contradicts maximality of $M$.
	Thus, it holds that $\Nout_{\oG_M}(\oY_n) \subseteq \oX_n$. Since it also holds that $\Nout_{\oG_M}(\oX_n) = N_{\oG}(\oX_n) = \oY_n$, we know that $\oQ$ has no outgoing edges in $\oG_M$. Since $\oG[\oX_I \cup \oY_H]$ is connected, and thus, by construction, also $\oG[\oX_n \cup \oY_n]$, it holds that $\oG_M[\oQ]$ is connected.
	
	It is well known that every finite directed graph has a tail strongly connected component. Let $\oS$ be any such component of $\oG_M[\oQ]$. As observed earlier, every $y_v \in \oS \cap \oY_n$ has an outgoing edge to some vertex in $\oX_n$. On the other hand, since we cannot apply \Cref{full:rr:VC_iso} and $\oY_n = \Nout_{\oG_M}(\oX_n)$, every $x_u \in \oS \cap \oX_n$ has an outgoing edge to some vertex in $\oY_n$. This way, $\oS$ contains at least one $\overline{X}$ and at least one $\overline{Y}$ vertex, i.e., it holds that $S_X, S_Y \neq \emptyset$.
	
	It remains to show that $S_X \cap S_Y = \emptyset$. Assume for contradiction that there is some $v \in S_X \cap S_Y$, i.e., it holds that $x_v \in \oS \cap \oX$ and $y_v \in \oS \cap \oY$. Let $u_1 = v$ and repeat from $i=1$ to $k$, such that $u_{k+1} \in \{u_1, \dots, u_k\}$: choose $u_{i+1}$ such that $x_{u_{i+1}}$ is the $M$-neighbor of $y_{u_i}$. Thus, $u_{i+1}$ is a neighbor of $u_i$ in $G$, and as a result, $y_{u_{i+1}}$ is an outgoing neighbor of $x_{u_i}$. As $\oS$ is strongly connected, we have visited all its vertices until $i=k$, so it holds that $\oS = \{x_{u_1}, y_{u_1}, \dots, x_{u_k}, y_{u_k}\}$. 
	Furthermore, it holds that $\oS = \oQ$, as otherwise by connectedness of $\oG_M[\oQ]$ there needs to be some edge $e$ between $\oQ \setminus \oS$ and $\oS$. If $e$ is outgoing from $\oS$, we contradict the fact that $\oS$ is a tail strongly connected component. Else $e$ is an ingoing edge for $\oS$, and since for every vertex in $\oS \cap \oX$ its only ingoing neighbor is also contained in $\oS$, this edge goes from some $x_p \in \oQ \setminus \oS$ to some $y_q \in \oS$. But then there also exists the edge from $x_q \in \oS$ to $y_p \in \oQ \setminus \oS$, which again contradicts the fact that $\oS$ has no outgoing edges. From $\oS = \oQ$ it follows that $\oX_I \subseteq \oS$, which together with $\oS$ being strongly connected implies that all $\oX_I$-vertices are matched by $M$. This leads to a contradiction: By earlier assumption, $M$ does not match $\oY_H$ into $\oX_I$; and at the same time, as $(I, H)$ is a crown, it holds that $|\oX_I| \geq |\oY_H|$ and that $\oX_I$ has no neighbors outside of $\oY_H$, implying that not all $\oX_I$-vertices are matched by $M$.
\end{proof}

\begin{lemma}\label{full:lem:crown_fast}
	\Cref{full:rr:VC_iso,full:rr:VC_crown} can be exhaustively applied in polynomial time.
\end{lemma}
\begin{proof}
	First, remove all isolated vertices contained in $R$. Then, set $\tilde{I}, \tilde{H} = \emptyset$ and repeatedly apply \Cref{full:rr:VC_crown} by constructing $\overline{G}_M$ and finding a tail strongly connected component $\overline{S}$ in $\overline{G}_M$ with $S_X \cap S_Y = \emptyset$ and $S_X, S_Y \neq \emptyset$, for which, by \Cref{full:lem:tscc-crown}, it holds that $(S_X, S_Y)$ is a crown in $G - (\tilde{I} \cup \tilde{H})$ with $S_X \subseteq R$.  If $E(S_Y, R)$ is more than a matching between $S_Y$ and $S_X$, delete the superfluous edges. After each such step, we set $\tilde{I}= \tilde{I}\cup S_X$, $\tilde{H} = \tilde{H} \cup S_Y$ and remove any new isolated vertices in $R$. Furthermore, the repeated computation of $\overline{G}_M$ can be avoided by simply removing $\overline{S}$ from $\overline{G}_M$, since $S_X$ and $S_Y$ are now part of $\tilde{I}$ and $\tilde{H}$. By \Cref{full:lem:crown-tscc}, if we cannot apply \Cref{full:lem:tscc-crown}, then there are no crowns left in $G - (\tilde{I} \cup \tilde{H})$ with $I \subseteq R$.
\end{proof}

\begin{lemma}\label{full:lem:big_crown_can_reduce}
	If there is a crown $(I, H)$ in $(G,B)$ with $I \subseteq R$ and $|I| > |H|$, then we can apply one among \Cref{full:rr:VC_iso,full:rr:VC_crown}.
\end{lemma}
\begin{proof}
	From the condition that $(I,H)$ is a crown it follows that $H = N(I)$ and there exists an $H$-saturating matching $M$ between $H$ and $I$. If $E(H,I) \setminus M \neq \emptyset$, we apply \Cref{full:rr:VC_crown}. Otherwise, we get $|I| > |H| = |M| = |E(H, I)|$ and thus at least one vertex in $I \subseteq R$ is isolated, so we apply \Cref{full:rr:VC_iso}.
\end{proof}

Now that we can remove all isolated vertices and crowns $(I, H)$ with $I \subseteq R$ and $|I| >|H|$ in polynomial time, we state that the resulting graphs do not have large independent sets. Since for \vertexcovervc all of $R$ is an independent set, it thus cannot be larger than $B$, which gives us the boundaried kernel. Note that neither the reduction rules, nor the polynomial-time application of them, nor the following lemma expect $G[R]$ to be an independent set, making this lemma also useful for other structural parameterization.

\begin{lemma}\label{full:lem:no_crown_small_IS}
	If $\alpha(G[R]) > \frac{1}{2} |V(G)|$, then we can apply one among \Cref{full:rr:VC_iso,full:rr:VC_crown}.
\end{lemma}
\begin{proof}
	Let $Y$ be a maximum independent set of $G[R]$, i.e., such that $|Y| = \alpha(G[R])$. Let $X$ be the neighborhood of $Y$, which is clearly a subset of $V(G) \setminus Y$. Note that since $|Y| = \alpha(G[R]) > \frac{1}{2}|V(G)|$, it thus holds that $|X| \leq |V(G)| - |Y| < \frac{1}{2}|V(G)| < |Y|$.
	
	If there exists an $X$-saturating matching into $Y$, then by \Cref{full:lem:big_crown_can_reduce} we can apply one of the reduction rules.
	Else we choose some maximal $\hat{X} \subseteq X$ such that $|\hat{X}| > |\hat{Y}|$ with $\hat{Y} = N(\hat{X}) \cap Y$, which exists by Hall's Theorem.
	Let $X^* = X \setminus \hat{X}$ and $Y^* = Y \setminus \hat{Y}$.
	Note that even though $X^*$ might have neighbors in $\hat{Y} = Y \setminus Y^*$, it holds that $N(Y^*) \subseteq X^*$, and by maximality of $\hat{X}$, also $N(Y^*) = X^*$.
	Further, as $Y^* \neq \emptyset$, if $(Y^*, X^*)$ is not a crown, it can only be due to absence of an $X^*$-saturating matching into $Y^*$. Then, again by Hall's Theorem, there exists some $\bar{X} \subseteq X^*$ with $|\bar{X}| > |N(\bar{X}) \cap Y^*|$, which contradicts maximality of $\hat{X}$.
	As a result, $(Y^*, X^*)$ must be a crown with $|Y^*| = |Y| - |\hat{Y}| > |Y| - |\hat{X}| > |X| - |\hat{X}| = |X^*|$. Again, by \Cref{full:lem:big_crown_can_reduce} we can apply one of \Cref{full:rr:VC_iso,full:rr:VC_crown}.
\end{proof}

\begin{theorem}\label{full:thm:VC_vc}
	The parameterized problem \vertexcovervc admits a polynomial boundaried kernelization with at most $2(|B|+k)$ vertices.
\end{theorem}
\begin{proof}
	By \Cref{full:lem:Bismod} assume that we are given boundaried graph $G_B$ with $G-B$ an independent set. Using \Cref{full:lem:crown_fast} apply \Cref{full:rr:VC_iso,full:rr:VC_crown} exhaustively in polynomial time, resulting in some boundaried graph $G'_B$. Observe that $G'$ is a subgraph of $G$, implying that $R' := V(G') \setminus B$ is an independent set, leading to $|R'| = \alpha(G'[R']) \leq \frac{1}{2}|V(G')|$ by \Cref{full:lem:no_crown_small_IS} and thus $|V(G')| = |B| + |R'| \leq 2|B|$. As a result, we can output $G'_B$ with $G'-B$ an independent set and $\Delta = 0$.
\end{proof}

\subsection{Vertex Cover[fvs]}
Using the polynomial kernel for \vertexcoverfvs by Bodlaender and Jansen \cite{DBLP:journals/mst/JansenB13} as a base, we describe our polynomial boundaried kernel for this problem. Main steps of the kernel by Bodlaender and Jansen are to make sure that $G-X$ (with $X$ a forest-modulator) has a perfect matching, to reduce the number of so called conflict structures in $G-X$, and to use that in order to bound the total size of $R$. We use \Cref{full:lem:no_crown_small_IS} for the first step, as it allows us to move vertices that are problematic for a perfect matching in $G[R]$ to $B$, with only a linear blow-up of the latter. Regarding the last two steps, we can mainly apply the same reduction rules as Bodlaender and Jansen. We only need to be careful to not delete boundary-vertices. Instead, we mark them by giving personal leaves and thus forcing them into a solution.

\begin{lemma}[Analogue of {\cite[Lemma 1]{DBLP:journals/mst/JansenB13}}]\label{full:lem:vcfvs_clean}
	Let $G_B \in \graphclass{forest}_B$ such that \Cref{full:rr:VC_iso,full:rr:VC_crown} cannot be applied. In polynomial time one can compute a set $B' \supseteq B$ of size at most $2|B|$, such that $G-B'$ has a perfect matching.
\end{lemma}
\begin{proof}
	With the Hopcroft-Karp algorithm we find a maximum matching $M$ of the forest $G[R]$, and let $I$ be the (independent) set of vertices not covered by $M$. Note that this way it holds that $|R| = 2|M| + |I|$. Since $G[R]$ is a forest and thus bipartite, by König's Theorem it has a minimum vertex cover of size $|M|$, and thus a maximum independent set of size $|R| - |M| = |M|+|I|$. From \Cref{full:lem:no_crown_small_IS} it follows that $|M|+|I| = \alpha(G[R]) \leq 1/2 |V(G)| = 1/2 (|B| + 2|M| + |I|)$, and hence $|I| \leq |B|$. For $B'$ we can thus choose the set $B \cup I$.
\end{proof}

Next, let us use additional definitions, which were first introduced by Bodlaender and Jansen \cite{DBLP:journals/mst/JansenB13} and generalized by Hols et al.~\cite{DBLP:journals/siamdm/HolsKP22}.

\begin{definition}[Blocking sets]
	For a graph $G$, a vertex set $Y \subseteq V(G)$ is called a \textit{blocking set} in $G$, if no minimum vertex cover of $G$ contains $Y$, i.e., for all solutions $S$ with $Y \subseteq S$ it holds that $|S| > \OPT_\VC(G)$. A blocking set $Y$ is called a \textit{minimal blocking set}, if no strict subset of $Y$ is also a blocking set. 
	
	We denote by $\beta(G)$ the size of the largest minimal blocking set in $G$. Further, for graph class $\cC$ let $b_\cC \coloneq \max_{G \in \cC}\beta(G)$ with $b_\cC = \infty$ if the minimal blocking set size of graphs in $\cC$ can be arbitrarily large.
\end{definition}

\begin{definition}[Chunks and conflicts]
	Let $G$ be a graph, $\cC$ a graph class and $X$ a $\cC$-modulator for $G$. A \textit{chunk} in $X$ is an independent vertex set $Z \subseteq X$ of size at most $b_\cC$. If $X$ is clear from context, we denote the set of chunks in $X$ as $\cX$. 
	
	With $F \subseteq V(G) \setminus X$, and $Z \in \cX$, let $\conf{F}{Z} \coloneq \OPT_\VC(G[F] - N_F(X)) + |N_F(X)| - \OPT_\VC(G[F])$ be the \textit{conflict caused by $Z$ on $F$}.
\end{definition}

Note that Bodlaender and Jansen have shown that $2$ is the maximum size of a minimal blocking set in a forest, i.e., it holds that $b_{\Cforest} = 2$ \cite[Lemma 3]{DBLP:journals/mst/JansenB13}. Thus, let $\cX \coloneq \{Z \subseteq B \mid G[Z] \text{ is an independent set and } |Z| \leq 2\}$ be the chunk-set for $G_B$. 
The following reduction rule is easily seen to be gluing safe, as it works with components in $G[R]$ for which we can take the locally optimum solution, without needing to worry about possible conflicts after gluing. This rule can be seen as a special case of \Cref{full:rr:vctdd}.

\begin{redrule}\label{full:rr:vcfvs_component}
	If $G[R]$ contains a connected component $T$ such that for every chunk $Z \in \cX$ it holds that $\conf{T}{Z} = 0$, then delete $T$ and increase $\Delta$ by $\OPT(T)$.
\end{redrule}

Next, let us bound the total number of conflicts that chunks in $\cX$ can cause on $R$.

\begin{lemma}\label{full:lem:vcfvs:badchunk}
	If for some chunk $Z \in \cX$ it holds that $\conf{G[R]}{Z} \geq |B|$, then for every $H_B$ there exists an optimum vertex cover $S$ of $G_B \oplus H_B$ such that $Z \cap S \neq \emptyset$.
\end{lemma}
\begin{proof}
	Let $H_B$ be fixed and let $S$ be a minimum vertex cover of $G_B \oplus H_B$. If it holds that $Z \cap S \neq \emptyset$, then we are done. So assume the opposite, i.e., that $Z \cap S = \emptyset$. For each connected component $T$ of $G[R]$, fix some minimum vertex cover $S_T$ of $T$. Let $\hat{S}$ be the union of all such $S_T$. Construct the set $S' \coloneq (S \setminus R) \cup B \cup \hat{S}$. It holds that $S'$ is a vertex cover of $G_B \oplus H_B$: (i) all edges in $H$ are covered, since $S' \cap V(H) \supseteq S \cap V(H)$; (ii) all edges in $G[R]$ are covered, since $S'$ contains all vertices in $\hat{S}$; (iii) all edges between $R$ and $B$ are covered, as $B \subseteq S'$. Since we assumed that $Z \cap S = \emptyset$ and it holds that $\conf{G[R]}{Z} \geq |B|$, it follows that $|S'| \leq |S| - |B| + |B| = |S|$, and thus $S'$ is also a minimum vertex cover of $G_B \oplus H_B$.
\end{proof}

Essentially, Bodlaender and Jansen used \Cref{full:lem:vcfvs:badchunk} to remove any chunk of size one with conflict at least $|B|$ in $R$, and, if the chunk has size two instead, to add an edge between the two vertices of $Z$. In our case, however, we cannot simply remove vertices in the boundary $B$, as they would get reintroduced after gluing. Instead, we mark such vertices by giving them special leaves. This is safe due to the well known reduction rule for \vertexcover, that neighbors of leaves can be safely taken to the solution.

Let $L = \{l_x \in V(G) \mid x \in B\}$ be the set of special leaves of boundary-vertices. At the beginning we set $L = \emptyset$. From now on, we set $R = V(G) \setminus (B \cup L)$. Note that at any point it will hold that $|L| \leq |B|$.

\begin{redrule}\label{full:rr:vcfvs_chunk1}
	If there is a vertex $x \in B$ without leaf $l_x \in L$, and such that $\conf{R}{x} \geq |B|$, then (i) add $l_x$ to $L$ with the edge $\{x, l_x\}$; and (ii) delete all edges between $x$ and $R$.
\end{redrule}
\begin{proof}[Proof of gluing safeness]
	Let $G'_B$ be the resulting graph. Fix any $B$-boundaried graph $H_B$. Let $S'$ be a minimum vertex cover for $G'_B \oplus H_B$. Since $x$ is adjacent to a leaf $l_x$ in $G'_B \oplus H_B$, we can assume that $x$ is contained in $S'$ (else $l_x \in S'$ and we can exchange it with $x$). Thus, all edges which exist in $G_B \oplus H_B$, but not in $G'_B \oplus H_B$ (namely, those between $x$ and $R$), are all covered by $S'$. 
	
	Now let $S$ be a minimum vertex cover for $G_B \oplus H_B$. By \Cref{full:lem:vcfvs:badchunk} we can assume that $S$ contains $x$. Thus, we know that $S$ also covers the only edge in $G'_B \oplus H_B$ that is not contained in $G_B \oplus H_B$ (namely $\{x, l_x\}$), and hence $S$ is also a minimum vertex cover for $G'_B \oplus H_B$.
\end{proof}

\begin{redrule}\label{full:rr:vcfvs_chunk2}
	If there is a chunk $\{x, y\} \in \cX$ with $\conf{R}{\{x, y\}} \geq |B|$, then add an edge between $x$ and $y$, and remove $\{x, y\}$ from the chunk set $\cX$.
\end{redrule}
\begin{proof}[Proof of gluing safeness]
	Let $G'_B$ be the resulting graph. For any $H_B$ it holds that $G_B \oplus H_B$ is a subgraph of $G'_B \oplus H_B$, so $\OPT_\VC(G_B \oplus H_B) \leq \OPT_\VC(G'_B \oplus H_B)$ is obvious.
	
	Now let $S$ be a minimum vertex cover for $G_B \oplus H_B$. By \Cref{full:lem:vcfvs:badchunk} we can assume that $S$ contains $x$ or $y$. Thus, we know that $S$ also covers the only edge in $G'_B \oplus H_B$  that is not contained in $G_B \oplus H_B$ (namely $\{x, y\}$), and hence $S$ is also a minimum vertex cover for $G'_B \oplus H_B$.
\end{proof}

As a last step, we need to bound the number of such local structures, on which the chunks do not cause any conflict, i.e., we know that any minimum vertex cover for the global graph (even after gluing) is also optimal locally.
For this, we use the notion of blockability which was also introduced by Bodlaender and Jansen.

\begin{definition}
	The vertex-pair $u,v \in R$ is called \textit{blockable}, if there is a chunk $Z \in \cX$ such that $u, v \in N(Z)$.
\end{definition}

Vertices $u,v$ being blockable by some chunk $Z$ can be seen in the sense that forbidding $Z$ from a vertex cover solution actually implies that we are forced to take $u$ or $v$ into the solution. It follows from the definition, that if $u,v$ are not blockable, then for any two $x \in N(u) \cap B$ and $y \in N(v) \cap B$ it holds that $x \neq y$ and $\{x, y\} \in E(G)$, i.e., no chunk $Z \in \cX$ is adjacent to both $x$ and $y$.

\begin{redrule}\label{full:rr:vcfvs_conf1}
	If $u, v \in R$ are not blockable and $\deg_{G[R]}(u), \deg_{G[R]}(v) \leq 2$, then do the following: (i) delete $u$ and $v$, and increase $\Delta$ by $1$; (ii) if $u$ has a neighbor $t \neq v$ in $R$, then make it adjacent to the old $B$-neighborhood of $v$; (iii) if $v$ has a neighbor $w \neq u$ in $R$, then make it adjacent to the old $B$-neighborhood of $u$; (iv) if both $t$ and $w$ exist, then add the edge $\{t, w\}$.
\end{redrule}
\begin{proof}[Proof of gluing safeness]
	Let $G'$ be the modified graph and fix some $H_B$. We show that $\OPT(G_B \oplus H_B) = \OPT(G'_B \oplus H_B) + 1$. Since we work with two graphs at the same time, namely those are $G_B \oplus H_B$ and $G'_B \oplus H_B$, let us fix that $N_B(u)$, resp., $N_B(v)$, is used for $N_{G_B \oplus H_B}(u) \cap B = N_G(u) \cap B$, resp., $N_{G_B \oplus H_B}(v) \cap B = N_G(v) \cap B$.
	
	($\OPT(G_B \oplus H_B) \geq \OPT(G'_B \oplus H_B) + 1$) Let us first note the differences in conditions that a vertex cover $S'$ for $G'_B \oplus H_B$ needs to fulfill, compared to what conditions a minimum vertex cover $S$ for $G_B \oplus H_B$ fulfills. That is, $S'$ does not need to cover any edges incident with $u$ or $v$, but it needs to cover: (i) the edges between $t$ and $N_B(v)$, if $t$ exists; (ii) the edges between $w$ and $N_B(u)$, if $w$ exists; and (iii) the edge $\{t, w\}$, if they both exist.
	
	Next, let us show that we can assume without loss of generality that $N_B(u) \cup \{v\} \subseteq S$. There is an edge between $u$ and $v$, so $S \cap \{u, v\} \neq \emptyset$. As $u, v$ are not blockable, there are no common vertices in $N_B(u)$ and $N_B(v)$, and every vertex in $N_B(u)$ is adjacent to every vertex in $N_B(v)$, so either one needs to be fully contained in $S$. The case $N_B(v) \cup \{u\} \subseteq S$ is symmetrical to $N_B(u) \cup \{v\} \subseteq S$. The cases where neither $N_B(u) \subseteq S$ nor $u \in S$, resp., neither $N_B(v) \subseteq S$ nor $v \in S$, trivially lead to uncovered edges.
	
	Now, if $t$ exists, we can assume w.l.o.g.\ that $t \in S$. Otherwise it holds that $u \in S$ and thus $(S \setminus \{u\}) \cup \{t\}$ is also a minimum vertex cover of $G_B \oplus H_B$. Hence, all new edges incident with $t$ are automatically covered by $S' \coloneq S \setminus \{v\}$. If $w$ exists, then either $w \in S$ or $N_{G_B \oplus H_B}(w) \subseteq S$. In the first case, all new edges incident with $w$ are automatically covered by $S'$. And the same also holds in the second case, since we assumed w.l.o.g.\ that $N_B(u) \subseteq S$ and $t \in S$.
	
	($\OPT(G_B \oplus H_B) \leq \OPT(G'_B \oplus H_B) + 1$) Again, let us first observe the new constraints on a vertex cover $S$ for $G_B \oplus H_B$, as compared to the constraints on a minimum vertex cover $S'$ for $G'_B \oplus H_B$. Those are that $S$ additionally needs to cover the edge $\{u, v\}$, the edges between $u$ and $N_B(u)$ and between $v$ and $N_B(v)$, and the edges $\{t, u\}, \{v,  w\}$ if $t$, resp., $w$, exist.
	
	Since $N_B(u) \cup N_B(v)$ induces a complete bipartite graph on $G'_B \oplus H_B$, we can assume w.l.o.g.\ that $N_B(u) \subseteq S'$. Let $S = S' \cup \{v\}$ and note that $S$ covers $\{u, v\}$ and the edges between $N_B(u)$ and $u$, and between $N_B(v)$ and $v$. If neither $t$ nor $w$ exist, we are done by choosing $S$. If both $t$ and $w$ exist, we can assume w.l.o.g.\ that $t \in S'$ (note that else it holds that $w \in S'$ and $N_B(v) \subseteq S'$, which is symmetrical to $t \in S'$ and $N_B(u) \subseteq S'$, so we choose $S' \cup \{u\}$), and thus $S$ covers also the new edges $\{t, u\}$ and $\{v, w\}$. If $w$ exists, but $t$ does not, then again $S$ covers the new edge $\{v, w\}$. Else it holds that $t$ exists and $w$ does not. If $t \in S'$, we are done by choosing $S$. Else it holds that $N_{G'_B \oplus H_B}(t) \subseteq S'$. Hence, in particular, it holds that $N_B(v) \subseteq S'$. In that case we choose $S' \cup \{u\}$ as the vertex cover for $G_B \oplus H_B$, and note that it indeed covers all new edges of $t$, $u$ and $v$.
\end{proof}

\begin{redrule}\label{full:rr:vcfvs_conf2}
	If there are distinct vertices $t, u, v, w$  in $R$ which satisfy $\deg_{G[R]}(u) = \deg_{G[R]}(v) = 3$, $N_R(t) = \{u\}$, $N_R(w) = \{v\}$ and $\{u, v\} \in E(G)$ such that the pairs $\{u, t\}, \{v, w\}$ and $\{t, w\}$ are not blockable, then do the following: (i) delete $t, u, v$ and $w$, and increase $\Delta$ by $2$; (ii) with $p$ being the third neighbor of $u$ in $R$, make $p$ adjacent to the old $B$-neighborhood of $t$; (iii) with $q$ being the third neighbor of $v$ in $R$, make $q$ adjacent to the old $B$-neighborhood of $w$.
\end{redrule}
\begin{proof}[Proof of gluing safeness]
	Let $G'$ be the modified graph and fix some $H_B$. We show that $\OPT(G_B \oplus H_B) = \OPT(G'_B \oplus H_B) + 2$.
	
	($\OPT(G_B \oplus H_B) \geq \OPT(G'_B \oplus H_B) + 2$) Let $S$ be a minimum vertex cover of $G_B \oplus H_B$. We do a case distinction. Let the first case be that both $u, v \in S$. If it also holds that $N_B(t), N_B(w) \subseteq S$, then $S \setminus \{u, v\}$ automatically covers all new edges in $G'_B \oplus H_B$. Otherwise assume w.l.o.g.\ that $N_B(w) \nsubseteq S$. Then it must hold that $w \in S$ and $N_B(v) \subseteq S$, as $\{v, w\}$ is not blockable. We then observe that $(S \setminus \{v\}) \cup \{q\}$ is also a vertex cover of same size as $S$ for $G_B \oplus H_B$, and this corresponds to the next case.
	
	Our second case is, w.l.o.g., that $u \in S$ and $v \not\in S$. Since $v \not\in S$, it follows that $w, q \in S$ and $N_B(v) \subseteq S$. If additionally it holds that $N_B(t) \subseteq S$, then $S \setminus \{u, w\}$ contains $q$ and covers the new edges of $p$. Otherwise it needs to hold that $t \in S$ and $N_B(u), N_B(w) \subseteq S$, as $\{u,t\}$ and $\{t, w\}$ are not blockable. Then we also know that $(S \setminus \{u, w\}) \cup \{p, v\}$ is a vertex cover of the same size as $S$ for $G_B \oplus H_B$, and this corresponds to the case that $v \in S$, $u \not\in S$ and $N_B(w) \subseteq S$, which is symmetrical to $u \in S, v \not\in S$ and $N_B(t) \subseteq S$. Thus we are done.
	
	($\OPT(G_B \oplus H_B) \leq \OPT(G'_B \oplus H_B) + 2$) Let $S'$ be a minimum vertex cover of $G'_B \oplus H_B$. Since $\{t, w\}$ are not blockable in $G$, it holds that $N_B(t) \cup N_B(w)$ induce a complete bipartite graph in $G$, thus also in $G'$ and $G'_B \oplus H_B$. This way, we determined that it must hold that $N_B(t) \subseteq S'$ or that $N_B(w) \subseteq S'$. W.l.o.g.\ assume the former. If the latter also holds, then $S' \cup \{u, v\}$ covers all new edges in $G_B \oplus H_B$ which are not contained in $G'_B \oplus H_B$. Namely, those are $\{p, u\}, \{t, u\}, \{u, v\}, \{v, w\}, \{v, q\}$, the edges between $u$ and $N_B(u)$, between $v$ and $N_B(v)$, between $t$ and $N_B(t)$, and between $w$ and $N_B(w)$.
	
	Otherwise it holds that $N_B(w) \nsubseteq S'$, and thus, since $q$ is adjacent to $N_B(w)$ in $G'$, and since $w$ and $v$ are not blockable in $G$, it must be the case that $q \in S'$ and that $N_B(v) \subseteq S'$. In such case it suffices to choose $S' \cup \{u, w\}$.
\end{proof}

Note that the number of conflicts that a chunk $Z \in \cX$ induces on a component in $G[R]$, as well as what structures are allowed to remain in $G[R]$, is exactly the same as in the kernel of Bodlaender and Jansen, since \Cref{full:rr:vcfvs_chunk1,full:rr:vcfvs_chunk2} still bound the sum of possible conflicts induced by $\cX$, and \Cref{full:rr:vcfvs_component,full:rr:vcfvs_conf1,full:rr:vcfvs_conf2} are exactly the same (neglecting that we increase the solution size offshift $\Delta$ instead of decreasing the sought solution size $k$).

\begin{definition}
	Define the number of \textit{active conflicts} induced on the forest $F$ by the chunks $\cX$ as $\confact{F}{\cX} \coloneq \Sigma_{Z\in \cX} \conf{F}{Z}$.
\end{definition}

\begin{lemma}[\cite{DBLP:journals/mst/JansenB13}, Observation 4]\label{full:lem:BoJa_active}
	Let $G_B$ be a boundaried graph such that $G-B$ is a forest. Let $L$ be a set of degree-$1$ vertices of $G-B$ that are adjacent to $B$. Assume that $G-(B \cup L)$ has a perfect matching, and that one cannot apply  \Crefrange{full:rr:vcfvs_component}{full:rr:vcfvs_conf2}. Let $R = V(G) \setminus \{B \cup L\}$. Then it holds that $\confact{R}{\cX} \leq |B|^2 + \binom{|B|}{2}|B|$. 
\end{lemma}

\begin{definition}[Conflict structures]
	Let $F$ be a forest with a perfect matching $M$.\\
	A \textit{conflict structure of type $A$} in $F$ is a pair of distinct vertices $\{v_1, v_2\}$ such that  $\{v_1, v_2\} \in M$ and $\deg_F(v_1) , \deg_F(v_2) \leq 2$.\\
	A \textit{conflict structure of type $B$} in $F$ is a path on four vertices $(v_1, v_2, v_3, v_4)$ such that $v_1$ and $v_4$ are leaves of $F$, and $\deg_F(v_2) = \deg_F(v_3) = 3$.
\end{definition}

\begin{lemma}[\cite{DBLP:journals/mst/JansenB13}, Lemma 7]\label{full:lem:BoJa_active_conf}
	Let $G_B$ be a boundaried graph such that $G-B$ is a forest. Let $L$ be a set of degree-$1$ vertices of $G-B$ that are adjacent to $B$. Assume that $G-(B \cup L)$ has a perfect matching, and that one cannot apply  \Crefrange{full:rr:vcfvs_component}{full:rr:vcfvs_conf2}. Let $R = V(G) \setminus \{B \cup L\}$. Let $\mathcal{S}$ be a set of vertex-disjoint conflict structures in $R$. Then $\confact{R}{\cX} \geq |\mathcal{S}|$.
\end{lemma}

\begin{lemma}[\cite{DBLP:journals/mst/JansenB13}, Theorem 1]\label{full:lem:BoJa_conf}
	Let $T$ be a tree with a perfect matching. There is a set $\mathcal{S}$ of mutually vertex-disjoint conflict structures in $T$ with $|\mathcal{S}| \geq |V(T)|/14$.
\end{lemma}

\begin{lemma}[\cite{DBLP:journals/mst/JansenB13}, Lemma 9]\label{full:lem:vcfvs_fast}
	We can exhaustively apply \Crefrange{full:rr:vcfvs_component}{full:rr:vcfvs_conf2} on $G_B$ in polynomial time.
\end{lemma}

\begin{lemma}\label{full:lem:vcfvs_size}
	Let $G_B$ be a boundaried graph and $L \subseteq V(G) \setminus B$ a set of degree-1 vertices that are adjacent to $B$, such that for $R \coloneq V(G) \setminus (B \cup L)$ it holds that $G[R]$ has a perfect matching. If one cannot apply any of \Crefrange{full:rr:vcfvs_component}{full:rr:vcfvs_conf2}, then it holds that $|R| \leq 14(|B|^2 + |B|^3)$.
\end{lemma}
\begin{proof}
	Let $\mathcal{S}$ be the set of conflict structures in $G[R]$ which is obtained by repeated use of \Cref{full:lem:BoJa_conf} on each component of $G[R]$. It holds that $|\mathcal{S}| \geq |R|/14$. By \Cref{full:lem:BoJa_active,full:lem:BoJa_active_conf} we get $|B'|^2 + \binom{|B|}{2}|B| \geq \confact{R}{\cX} \geq |R|/14$, i.e., it holds that $|R| \leq 14(|B|^2 + |B|^3)$.
\end{proof}

	\begin{theorem}
	The parameterized problem \vertexcoverfvs admits a polynomial boundaried kernelization with at most $\Oh((|B| + k)^3)$ vertices.
\end{theorem}
\begin{proof}
		Let boundaried graph $G_B$ and feedback vertex set $X$ of $G$ be given. Using \Cref{full:rr:VC_iso,full:rr:VC_crown} and \Cref{full:lem:vcfvs_clean}, compute a set $B' \supset B \cup X$ of size at most $2(|B| + k)$, such that $G - B'$ has a perfect matching. Note that $G - B'$ also remains a forest. Using \Cref{full:lem:vcfvs_fast}, exhaustively apply \Crefrange{full:rr:vcfvs_component}{full:rr:vcfvs_conf2} on $G_{B'}$ in polynomial time. Let $G'_{B'}$ be the resulting graph, $L$ the set of vertices introduced by the repeated use of \Cref{full:rr:vcfvs_chunk1}, $\Delta$ the resulting offset and $R = V(G') \setminus (B' \cup L)$. By \Cref{full:lem:vcfvs_size} it  holds that $|R| \leq 14(|B'|^2 + |B'|^3)$. By gluing safeness of the applied reduction rules, it holds that $G_{B'} \equiv_{\VC, B'} G'_{B'}$ and since $B \subseteq B'$, we know by \Cref{full:lem:BD} that for all $H_B$: $\OPT(G_B \oplus H_B) = \OPT(G'_B \oplus H_B) + \Delta$. As a result, our boundaried kernelization algorithm outputs $G'_B$ with at most $\Oh((|B|+k)^3)$ vertices, its feedback vertex set $B'$ of size $2(|B|+k)$ and offset $\Delta$.
\end{proof}

\subsection{Feedback Vertex Set[fvs]}
The state of the art vertex-quadratic kernelization for $\feedbackvertexset[k]$ was given by Thomassé \cite{DBLP:journals/talg/Thomasse10}. It mainly works through handling vertices of low degrees, as well as certain flowers, i.e., sets of cycles intersecting at exactly one vertex, which are shown to exist at vertices of high degree, and recognizing that this bounds the total number of vertices by a quadratic function on the sought solution size, unless certainly given a NO-instance. We have no sought solution size at hand, but we have a nice structure with $G[R]$ only consisting of trees as components. Using similar reduction rules to those given by Thomassé, we force the low-degree vertices of $R$ (i.e., leaves and non-branching vertices) to be adjacent to $B$ on one hand, and limit the maximum degree of the $B$-vertices by $\Oh(|B|)$ on the other hand. Together, this bounds the allowed number of vertices in $R$ by $\Oh(|B|^2)$.

Throughout this section we allow loops, i.e., an edge from a vertex to itself, and double edges, i.e., that an edge $\{u, v\}$ exists up to twice. If an edge already exists twice and we would try to add a third copy, we instead discard the new copy.

First, let us handle $R$-vertices of low degree, $B$-vertices with loops and vertices $x \in B$ that admit an $x$-flower of order greater than $|B|$, i.e., $x$ is the intersection point of at least $|B|+1$ many otherwise disjoint cycles. The gluing safeness of these rules is easy to see, since $R$-vertices of degree $0$ and $1$ are not contained in any cycle, even after gluing; while $B$-vertices with loops and high-order flowers need to be contained in any solution, hence additional edges at these vertices are of no interest.

\begin{redrule}\label{full:rr:fvs_deg01}
	If exists vertex $v \in R$ with degree $0$ or $1$, then delete $v$.
\end{redrule}
\begin{redrule}\label{full:rr:fvs_deg2}
	If some vertex $v \in R$ is incident with exactly two edges $\{v, u\}$ and $\{v, w\}$ (possibly  with $u = w$), then delete $v$ and add the edge $\{u, w\}$. 
\end{redrule}
\begin{redrule}\label{full:rr:fvs_loop_B}
	If for some vertex $x \in B$ there is a loop at $x$ and $x$ is incident with further edges, then delete all those other edges incident with $x$.
\end{redrule}
\begin{redrule}\label{full:rr:fvs_flower}
	If for some vertex $x \in B$ there is an $x$-flower $F$ in $G$ of order $p > |B|$, then add a loop at $x$ and delete all other edges incident with $x$.
\end{redrule}

Next, we work with vertices $x \in B$ that admit a special structure, to which we first want to give a little intuition: Assume that for $x$ there exists a vertex set $X \in V \setminus \{x\}$ and set $\cC$ of connected components in $G-(X \cup\{x\})$ not intersecting $B$, such that there is exactly one edge between $x$ and any component $C \in \cC$, and for any subset $Z \subseteq X$ there are at least $2|Z|$ components in $\cC$ that are adjacent to $Z$. The last point implies, by Hall's theorem, that we can assign to each vertex in $X$ its distinct two adjacent components of $\cC$. This way, we obtain an $x$-flower of order $|X|$, in which each cycle goes from $x$ through some component in $\cC$ to some vertex in $X$, and back to $x$ through another component in $\cC$, with each cycle visiting its distinct two components. Since each component of $\cC$ has no neighbors outside of $X \cup \{x\}$ (even after gluing, as it does not intersect $B$) and is itself a tree, there exists an optimum solution which contains either $x$ or all of $X$. We refer to the work of Thomassé \cite{DBLP:journals/talg/Thomasse10} for the formal proof of (gluing) safety, since by choice of $V(C) \cap B = \emptyset$ for every $C \in \cC$, we make sure that this structure preserves even after gluing.

\begin{redrule}\label{full:rr:fvs_gallai}
	Assume there is a vertex $x \in V$, a vertex set $X \subseteq V \setminus \{x\}$ and a set $\cC$ of connected components of $G - (X \cup \{x\})$ (not necessarily all of the connected components) such that: (i) for every $C \in \cC$ it holds that $V(C) \cap B = \emptyset$; (ii) there is exactly one edge between $x$ and every $C \in \cC$; (iii) every $C \in \cC$ induces a tree; (iv) for every subset $Z \subseteq X$, the number of components in $\cC$ having some neighbor in $Z$ is at least $2|Z|$. Then add a double edge between $x$ and every vertex in $X$, and delete the edges between $x$ and the components in $\cC$.
\end{redrule}

As a last step, we show that as long as some $B$-vertex has at least $5|B|$ neighbors in $R$, we can find a reduction rule to apply. Both finding the rule and applying it can be done in polynomial time.

\begin{lemma}[\cite{DBLP:journals/talg/Thomasse10}, Corollary 2.2]\label{full:lem:max_flower}
	Let $G$ be a graph and $x$ be a vertex of $G$ which is not incident with a loop. The maximum order of an $x$-flower is equal to the minimum of $|X| + \Sigma_{C \in \cC}\lfloor \frac{e(x, C)}{2} \rfloor$, where $X$ is a subset of vertices, $\cC$ is the set of components of $G - (X \cup \{x\})$, and $e(x, C)$ is the number of edges between $x$ and vertices in $C$.
	Moreover, the set $X$ can be computed in time polynomial in the size of $G$.
\end{lemma}

\begin{lemma}[\cite{DBLP:journals/talg/Thomasse10}, Theorem 2.4]\label{full:lem:expansion}
	Let $G$ be a nonempty bipartite graph with bipartition $V = X \dot\cup Y$ with $|Y| \geq 2|X|$ and such that every vertex of $Y$ has at least one neighbor in $X$. Then there exist nonempty sets $X' \subseteq X$ and $Y' \subseteq Y$ such that $N(Y') = X'$ and such that every $Z \subseteq X'$ has at least $2|Z|$ neighbors in $Y'$. In addition, $X'$ and $Y'$ can be computed in time polynomial in the size of $G$.
\end{lemma}

\begin{lemma}\label{full:lem:fvs_deg}
	If for some $x \in B$ it holds that $|N_R(x)| \geq 5|B|$, then one can recognize applicability of one among \Crefrange{full:rr:fvs_deg01}{full:rr:fvs_gallai} in time polynomial in the size of $G$.
\end{lemma}
\begin{proof}
	Recognizing \Crefrange{full:rr:fvs_deg01}{full:rr:fvs_loop_B} is simple to do, so assume that none of them can be applied.
	In particular, there is no loop at $x$, and thus we make use of \Cref{full:lem:max_flower} in order to find a set of vertices $X \subseteq V \setminus \{x\}$, such that the maximum order $p$ of an $x$-flower is equal to $|X| + \Sigma_{C \in \cC} \lfloor \frac{e(x, C)}{2} \rfloor$, where $\cC$ is the set of components of $G - (X \cup \{x\})$, and $e(x, C)$ is the number of edges between $x$ and $C$. If $p > |B|$, then apply \Cref{full:rr:fvs_flower}. Otherwise, let $\cC' \subseteq \cC$ be the set of  components of $\cC$ that have more than one edge to $x$. We denote by $e'$ the total number of edges between $x$ and the components of $\cC'$. Note that by choice of $X$, it holds that $|X| + e'/3$ is at most $p$. Since we could not apply \Cref{full:rr:fvs_flower}, there is no $x$-flower of order $|B|+1$, so we get $3|X| + e' \leq 3p \leq 3|B|$ and thus also $|X| \leq |B|$. Furthermore, there are at most $|B|$-many components in $\cC \setminus \cC'$ that contain $B$-vertices.
	
	Together, the number $c$ of components in $\cC$ that are disjoint from $B$ and linked to $x$ with exactly one edge is at least as high as the minimum degree $5|B|$ of $x$, minus $|X|$ for the neighbors of $x$ in $X$, minus $e'$, minus the number of components that contain $B$-vertices, and minus the number of $X$-vertices that are adjacent by double edges to $x$. Hence, it holds that $c \geq 5|B| - |X| - e' - 2|B| = 3|B| - |X| - e'$. Together with $3|B| \geq 3|X| + e'$ this gives us $c \geq 2|X|$.
	Let $Y$ be the set of these at least $2|X|$-many components in $\cC$, which induce trees, do not intersect $B$ and are incident with exactly one edge to $x$ each.
	
	Let $A$ be a bipartite graph with vertex bipartition $V(A) = X \dot\cup Y$. Add an edge between $v \in X$ and $C \in Y$ if and only if there exists an edge between $v$ and $C$ in $G$. Let us fix any $C \in Y$ and show that it is adjacent to some $v \in X$ in $A$. Since $C$ is a subgraph of $R$, and because we cannot apply \Cref{full:rr:fvs_deg01} or \Cref{full:rr:fvs_deg2}, every vertex in $C$ has degree at least $3$ in $G$, including the leaves of $C$. Since every leaf of $C$ is incident with at most one edge to $x$, and has at most one edge inside of $C$, it needs to be also adjacent to $X$, and thus $C$ is adjacent to some vertex $v \in X$ in $A$. Together with the fact that $|Y| \geq 2|X|$, we can apply \Cref{full:lem:expansion} and find nonempty sets $X' \subseteq X$ and $Y' \subseteq Y$, such that $N_A(Y') = X'$, and such that every subset $Z \subseteq X'$ has at least $2|Z|$ neighbors in $Y'$. Summing up, we have that every element $C$ of $Y'$ is a connected component of the graph $G \setminus (X' \cup \{x\})$, and $C$ is adjacent to $x$ through exactly one edge, $C$ does not intersect $B$ and induces a tree on $G$, and that every subset $Z$ of $X'$ is adjacent to at least $2|Z|$ components in $Y'$. Hence, we can apply \Cref{full:rr:fvs_gallai}.
\end{proof}

\begin{theorem}\label{full:thm:PBK_FVS_fvs}
	The parameterized problem \feedbackvertexsetfvs admits a polynomial boundaried kernelization with at most $\Oh((|B|+k)^2)$ vertices.
\end{theorem}
\begin{proof}
	Assume by \Cref{full:lem:Bismod} that we are given a boundaried graph $G_B$ such that $G-B$ is a forest.
	Note that each among \Crefrange{full:rr:fvs_deg01}{full:rr:fvs_gallai} decreases the value $n+s$, where $n$ and $s$ are the number of vertices and simple edges of $G$, respectively. This gives a bound linear on $|G|$ for the number of times we apply these reduction rules. 
	
	Let $G'_B$ be the result after exhaustively applying \Crefrange{full:rr:fvs_deg01}{full:rr:fvs_flower} and \Cref{full:lem:fvs_deg}, and let $R' = V(G') \setminus B$. Observe, that there are less than $5|B|^2$ many edges between $B$ and $|R'|$. Since we cannot apply \Cref{full:rr:fvs_deg01} and \Cref{full:rr:fvs_deg2}, every vertex in $R'$ has degree at least three in $G'$. Hence, each leaf in $R'$, as well as each non-branching inner vertex of $R'$, has at least one $B$-neighbor, which means that there are less than $5|B|^2$-many leaves and non-branching vertices in $R'$. Further, the number of branching vertices in a forest is upper bounded by the number of leaves. Together, this yields a size bound of at most $\Oh(|B|^2)$ for $R'$, and thus for the size of $V(G') = B \cup R'$. The output is thus $G'_B$ with $B$ being a modulator to forest, and $\Delta = 0$.
\end{proof}

\subsection{Long Cycle[vc] and related problems}
Our next positive results regarding polynomial boundaried kernelization are for problems that are not formulated in terms of vertex-deletion to some graph class $\cC$, namely \longcyclevc, \longpathvc, \hamcyclevc, \hampathvc, \hamcycledeg, \hampathdeg. These problems admit polynomial kernelization, as was shown by Bodlaender et al.~\cite{DBLP:journals/tcs/BodlaenderJK13}. Mainly, our argumentation will work very similarly to theirs, however, we are not able to easily reduce \longpathvc to \longcyclevc by simply adding a global vertex which is adjacent to all original vertices, as we can only  modify $G_B$, and thus such a vertex would not be adjacent to all vertices of $H$ in the glued instance $G_B \oplus H_B$. Still, a very similar argumentation to the kernelization for \longcyclevc is able to help us out.

We first introduce a result of Bodlaender et al.\ which will be our main tool.

\begin{lemma}[\cite{DBLP:journals/tcs/BodlaenderJK13}, Theorem 2]\label{full:lem:lc_matching}
	Let $G = (X \dot\cup Y, E)$ be a bipartite graph. Let $M \subseteq E(G)$ be a maximum matching in $G$. Let $X_M \subseteq X$ be the set of vertices in $X$ that are endpoints of an edge in $M$. Then, for each $Y' \subseteq Y$, if there exists a matching $M'$ in $G$ that saturates $Y'$, then there exists a matching $M''$ in $G[X_M \cup Y]$ that saturates $Y'$.
\end{lemma}

For $G_B$, construct the bipartite graph $A$, with one side of the vertex set of $A$ being $R$, and the other consisting of pairs of distinct vertices in $B$. There is an edge between $v \in R$ and $\{p, q\} \subseteq B$, if and only if $v$ is adjacent to both $p$ and $q$. Let $M_A$ be a maximum matching in $A$ and let $J_A \subseteq R$ be the set of those vertices that are matched by $M_A$.
If $G$ contains cycles of size $4$ with exactly two vertices in $R$, then fix one such cycle $C'$.
If such $C'$ exists, then let $K$ contain the two $R$-vertices of $C'$, and $B_K$ the two $B$-vertices of $C'$. If no such $C'$ exists, then let $K = B_K = \emptyset$.

\begin{redrule}[Analogue of {\cite[Reduction Rule 1]{DBLP:journals/tcs/BodlaenderJK13}}]\label{full:rr:lc}
	Delete all vertices in $R \setminus (J_A \cup K)$ from $G_B$.
\end{redrule}
\begin{proof}[Proof of gluing safeness]
	Let $G'_B$ be the resulting boundaried graph. Let $H_B$ be any boundaried graph to be glued to $G_B$, resp., to $G'_B$. Since $G'_B \oplus H_B$ is a subgraph of $G_B \oplus H_B$, it clearly holds that $\OPT_\LC(G_B \oplus H_B) \geq \OPT_\LC(G'_B \oplus H_B)$. Let us also show that the reverse is true.
	
	Let $C$ be a cycle in $G_B \oplus H_B$. Clearly, whenever a vertex $v \in V(R)$ is visited by $C$, it holds that $C$ visited some vertices in $B$ before and after $v$. Let $v_1, \dots, v_r$ be all the vertices of $R$ contained in $C$, and let $p_i, q_i$ be the predecessor and successor of $v_i$ in $C$, respectively. 
	
	Let us first inspect the case that $C$ has length $4$ and $r = 2$, i.e., $C = (p_1, v_1, q_1, v_2, p_1)$.
	Since $C$ witnesses the existence of a size-$4$ cycle with exactly two vertices in $R$, it holds true that exists a cycle $C'$ of size $4$, and since its vertices remain in $G'_B$, we know that $C'$ exists in $G'_B \oplus H_B$. 
	
	Now we can assume that there are no $i,j \in [r]$ with $i \neq j$ and $\{p_i, q_i\} = \{p_j, q_j\}$. 
	Let $W =  \{\{p_i, q_i\} \mid i \in [r]\}$. Note that there exists a $W$-saturating matching in $A$, namely matching each $\{p_i, q_i\}$ with $v_i$. By \Cref{full:lem:lc_matching}, there also exists a $W$-saturating matching $M'$ in $A[W \cup J_A]$.
	
	For every $i \in [r]$, let $v'_i$ denote the vertex to which $\{p_i, q_i\}$ is matched by $M'$. It is easy to see that we may replace each $v_i$ on $C$ by $v'_i$ since $v'_i$ is adjacent to $p_i$ and $q_i$ in $G$, obtaining a cycle $C^*$ which intersects $I$ only at $J_A$. As all pairs $\{p_i, q_i\}$ are different, no vertex $v'_i$ is required twice. Hence $C^*$ is also a cycle of $G'_B \oplus H_B$ of same size as $C$.
\end{proof}

Constructing $A$ and computing $M_A$, as well as applying \Cref{full:rr:lc} clearly can all be done in polynomial time. Moreover, $J_A$ has size at most $|B|^2$, and thus after application of \Cref{full:rr:lc}, $R$ contains at most $|B|^2 + 2 \in \Oh(|B|^2)$ vertices. Together with \Cref{full:lem:Bismod} we get the following result.

\begin{theorem}
	The parameterized problem \longcyclevc admits a polynomial boundaried kernelization with at most $\Oh((|B|+k)^2)$ vertices.
\end{theorem}

For \parameterizedproblem{Long Path}{vc} we construct a different auxiliary bipartite graph, let it be denoted by $D$. Let the vertex set of $D$ consist of $R$ on one side and let the other side of $V(D)$ be the union of $B$ and the set of pairs of distinct vertices in $B$. We compute a maximum matching $M_D$ in $D$ and denote by $J_D \subseteq I$ the set  of those vertices that are matched by $M_D$.

\begin{redrule}\label{full:rr:lp}
	Delete all vertices in $R \setminus J_D$ from $G_B$.
\end{redrule}
\begin{proof}[Proof of gluing safeness]
	Let $G'_B$ be the resulting boundaried graph. Let $H_B$ be any boundaried graph to be glued to $G_B$, resp., to $G'_B$. Since $G'_B \oplus H_B$ is a subgraph of $G_B \oplus H_B$, it clearly holds that $\OPT_\LP(G_B \oplus H_B) \geq \OPT_\LP(G'_B \oplus H_B)$. Let us also show that the reverse is true.
	
	Let $P$ be a path in $G_B \oplus H_B$. Let $R'$ be the set of $R$-vertices visited by $P$. We denote the vertices of $R'$ in the order they appear in $P$, i.e., $R' = \{v_1, \dots, v_r\}$, and for every $i, j \in [r]$ with $i < j$ it holds that $P$ visits $v_i$ before $v_j$.
	
	Assume that every $R'$-vertex $v_i$ has a predecessor $p_i$ and a successor $q_i$ in $P$, both of which are obviously contained in $B$. Let $W =  \{\{p_i, q_i\} \mid i \in [r]\}$. Note that there exists a $W$-saturating matching in $D$, namely matching each $\{p_i, q_i\}$ with $v_i$. By \Cref{full:lem:lc_matching}, there also exists a $W$-saturating matching $M'$ in $D[W \cup J_D]$. For every $i \in [r]$, let $v'_i$ denote the vertex to which $\{p_i, q_i\}$ is matched by $M'$. It is easy to see that we may replace each $v_i$ on $P$ by $v'_i$ since $v'_i$ is adjacent to $p_i$ and $q_i$ in $G$, obtaining a path $P'$ which intersects $I$ only at $J_D$. As all pairs $\{p_i, q_i\}$ are different, no vertex $v'_i$ is required twice. Hence $P'$ is also a path of $G'_B \oplus H_B$ of same size as $P$.
	
	Now assume that $P$ starts with $v_1$ and ends with $v_r$, i.e., $v_1$ has only a successor $q_1$ in $P$, and $v_r$ only a predecessor $p_r$. Then let $W = \{q_1\} \cup \{p_r\} \cup \{\{p_i, q_i\} \mid i \in \{2, \dots, r-1\}\}$. Again, $W$ has a saturating matching in $D$, namely matching $v_1$ with $q_1$, and $v_r$ with $p_r$, and for each $i \in \{2, \dots, r-1\}$, matching $\{p_i, q_i\}$ with $v_i$. Using \Cref{full:lem:lc_matching}, we get a $W$-saturating matching $M'$ in $D[W \cup J_D]$. Denote by $v'_1$ the vertex to which $q_1$ is matched by $M'$, by $v'_r$ the vertex to which $p_r$ is matched by $M'$, and for each $i \in \{2, \dots, r-1\}$, denote by $v'_i$ the vertex to which $\{p_i, q_i\}$ is matched by $M'$. Again, it is easy to see that we may replace each $v_i$ on $P$ by $v'_i$ and obtain a path $P'$, which intersects $I$ only at $J_D$ and is thus also a path of $G' \oplus H_B$.
	
	There are two other cases: (i) $P$ starts with $v_1$ but does not end with $v_r$; and (ii) $P$ does not start with $v_1$ but ends with $v_r$. It is straightforward to handle these two cases in a similar manner as the first two cases were handled.
\end{proof}

Again, constructing $D$ and computing $M_D$, as well as applying \Cref{full:rr:lp} clearly can all be done in polynomial time. Moreover, $J_D$ has size at most $|B|^2 + |B|$, and thus after application of \Cref{full:rr:lp}, $R$ contains at most $|B|^2 + |B| \in \Oh(|B|^2)$ vertices. Together with \Cref{full:lem:Bismod} we get the following result.

\begin{theorem}
	The parameterized problem \longpathvc admits a polynomial boundaried kernelization with at most $\Oh((|B|+k)^2)$ vertices.
\end{theorem}

\begin{theorem}\label{full:thm:HC/P_vc}
	The parameterized problems \hamcyclevc, \hampathvc each admit a polynomial boundaried kernelization with at most $\Oh(|B| + k)$ vertices.
\end{theorem}
\begin{proof}
	Assume that $G-B$ is an independent set and apply \Cref{full:lem:Bismod} to lift this assumption and get a polynomial boundaried kernelization with $\Oh(|B|+k)$ vertices.
	
	If $|V(G)| \leq 2|B|+1$, we are done. Else for any $H_B$ there is no Hamiltonian cycle and no Hamiltonian path on $G_B \oplus H_B$. Assume for contradiction that such a cycle or path exists, and denote it by $P$. In particular, $P$ should visit all vertices in $R$. Let us be given the ordering of the $R$-vertices in which they are visited by $P$, i.e., let $R = \{v_1, \dots, v_r\}$ with $r=|B|$. If $P$ is a cycle, then simply choose any of the $R$-vertices as the first one. By definition, for each $i \in [r-1]$ there needs to be a path from $v_i$ to $v_{i+1}$, and since $R$ is an independent set, this path needs to visit at least one $B$-vertex. As the inner vertices of these paths need to be disjoint for any $i_1, i_2 \in [r-1]$ with $i_1 \neq i_2$, it follows that $P$ needs to visit at least $|R|-1$ different $B$-vertices. By assumption, it holds that $|R| > |B|+1$, i.e., $|B| < |R|-1$, and as a result, our assumed path/cycle $P$ cannot exist.
\end{proof}

\begin{theorem}
	Each among the parameterized problems \parameterizedproblem{Hamiltonian Cycle}{$\#v, \text{deg}(v) \neq 2$} and \parameterizedproblem{Hamiltonian Path}{$\#v, \text{deg}(v) \neq 2$} admits a polynomial boundaried kernelization with at most $\Oh(|B| + k)$ vertices.
\end{theorem}
\begin{proof}
	Let us be given a boundaried graph $G_B$ and let $X$ be the set of vertices in $G$ with degree different than two. Let $R_X = V(G - (B \cup X))$. It is easy to see the gluing safeness of contracting any two adjacent vertices $u, v$ in $R_X$ to some vertex $w_{u,v}$: Computing a solution for $G_B \oplus H_B$ from one for $G'_B \oplus H_B$ can be done by replacing $w_{u,v}$ by $u, v$ in the correct order; for the other direction we can assume without loss of generality, that $v$ appears directly after $u$ in the solution, which lets us replace $u, v$ by $w_{u,v}$. Exhaustive application of this reduction rule leaves us with $B \cup X$ being a vertex cover. Hence apply the polynomial boundaried kernelization from \Cref{full:thm:HC/P_vc} to obtain gluing-equivalent graph $\hat{G}_B$ and offset $\Delta$. Further, it still holds that $X$ contains all vertices in $\hat{G}$ with degree different than two.
\end{proof}


\section{Lower bounds}\label{full:section:lowerbound}
Despite the relatively good behavior of the polynomial kernels used to obtain boundaried counterparts in \Cref{full:section:upperbound}, one cannot expect this to be the case for all problems. Among such, we show \clustereditingparamce and \treedeletiontds to not admit any boundaried kernelization, although polynomial kernels are known for $\clusterediting[k]$ and $\treedeletion[k]$ \cite{DBLP:journals/mst/GrammGHN05,DBLP:journals/algorithmica/CaoC12,DBLP:journals/siamdm/GiannopoulouLSS16}. Note that (polynomial) kernelization for $\Pi[k]$ implies such also for $\Pi[\Pi]$ with $\Pi$ being a pure graph minimization problem: given $(G, k, s)$, if $\Pi(G, s) \leq k$, then answer YES, else a (polynomial) bound on $k$ also implies such on $\Pi(G, s)$.

Most of our lower bound results are obtained through excluding finite integer index, which was used as a sufficient condition for effective protrusion replacement earlier \cite{DBLP:journals/jacm/BodlaenderFLPST16}. Furthermore, we also show a ``too large'' integer index in order to unconditionally rule out a polynomial boundaried kernelization for the problem \dominatingsetvc, which is consistent with \Cref{full:lem:PBK-PK} and an existing conditional exclusion of a polynomial kernel for the problem \cite{DBLP:conf/icalp/DomLS09}. A similar result is obtained by Jansen and Wulms~\cite{DBLP:journals/dam/JansenW20} who showed that \vertexcover and \dominatingset do not admit a single-exponential number of gluing-equivalence classes, even with the restriction to planar graphs with treewidth at most $|B| + \Oh(1)$. Together with \Cref{full:lem:BK-index} their result excludes polynomial boundaried kernelization for $\vertexcover[\text{tw} + \modto{\text{planar}}]$ and $\dominatingset[\text{tw} + \modto{\text{planar}}]$.

\begin{figure}[t]
	\centering
	\quad
	\begin{subfigure}[t]{0.39\textwidth}
		\centering
		\includegraphics[width=\textwidth]{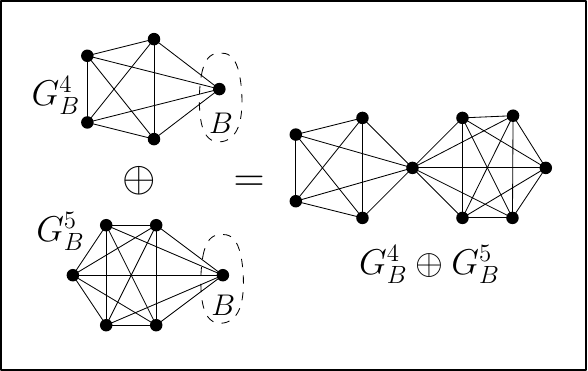}
	\end{subfigure}\quad
	\hfill
	\begin{subfigure}[t]{0.51\textwidth}
		\centering
		\includegraphics[width=\textwidth]{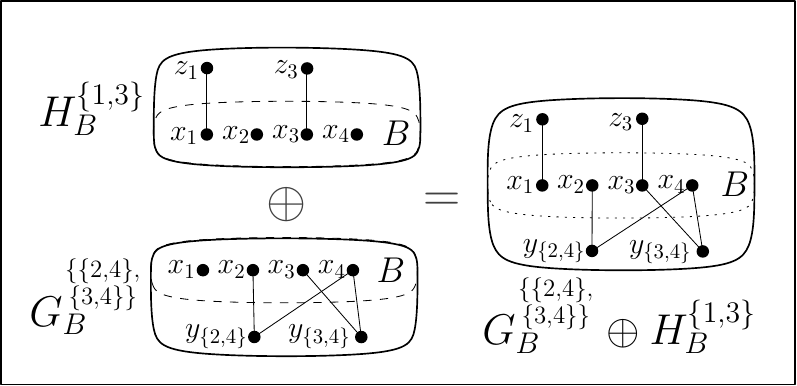}
	\end{subfigure}\quad \quad
	\caption{Examples of graphs defined in proofs for Lemmas \ref{full:lem:noFII:CE_complete} (left) and \ref{full:lem:noSEII:DS_vc} (right).}\label{full:fig:ds_ce}
\end{figure}

\subsection{Cluster Editing}

\begin{lemma}\label{full:lem:noFII:CE_complete}
	\clusterediting does not have finite integer index, even on the class of complete graphs.
\end{lemma}
\begin{proof}
	Choosing our boundary $B$ to consist of exactly one vertex $x$, we show that the equivalence relation $\equiv_{\CE, B}^{\graphclass{complete}}$, from now on denoted simply $\equiv$ throughout this proof, has infinitely many equivalence classes, by giving an infinite set of boundaried graphs $G_B^i, i \in \mathbb{N}$ with $G^i$ being complete, such that for any distinct $i, j \in \mathbb{N}$ it holds that $G_B^i \not\equiv G_B^j$. Namely, we define $G_i$ as a complete graph on $i+1$ vertices, with the boundary vertex $x$ being any one of the vertices. It is easy to see that these graphs are all non-isomorphic. The boundaried graph $G_B^0$ consists solely of the vertex $x$.
	
	For any $i \in \mathbb{N}$ it holds that $G^i_B \oplus G^0_B = G^i$, which is a clique on $i+1$ vertices, and hence the optimum \clusterediting solution for this graph has size zero. Similarly, we show that for any $i, j \in \mathbb{N}$ with $i > j$ it holds that $\OPT_{\CE}(G^i_B \oplus G^j_B) = j$. Without loss of generality we will assume that the vertex sets of $G^i$ and $G^j$ intersect at $x$ only. Let $G = G^i_B \oplus G^j_B$ and denote $E(G)$ by $E$. Observe that $G$ consists of two cliques, one of size $i+1$ and the other of size $j+1$, intersecting exactly at $x$ only. Let $C^i = V(G^i - x)$ and $C^j  = V(G^j - x)$. We can choose $S = E(x, C^j) = \{\{x, v\} \mid v \in C^j\}$ as a solution for $G$ of size $j$. One easily sees, that $G \triangle S$ consists of two cliques, $C^i \cup  \{x\}$ and $C^j$. Assume that there is some better solution $S'$ of size less than $j$. It is clear that $E(x, C^j), E(x, C^i) \nsubseteq S'$, as otherwise we would have $|S'| \geq j$. Thus in $G \triangle S'$, vertex $x$ is still adjacent to some vertices in both $C^i$ and $C^j$. Let $v \in C^j$ be a vertex in $C^j$ that remains adjacent to $x$ in $G \triangle S'$. Similarly, let $U \subseteq C^i$ be the (maximal) set of all vertices in $C^i$ that remain adjacent to $x$ in $G \triangle S'$. As $(u, b, v)$ with  $u \in U$ forms an induced $P_3$ in $G$, it follows that we need the edge $\{u, v\}$ as well in $S'$, to make it a triangle. By maximality of $U$, any vertex in $C^i \setminus U$ needs to be disconnected from $x$ in $G \triangle S'$, and thus all edges in $E(x, C^i \setminus U)$ need to be in $S'$. Counting the number of edges already found to be in $S'$, we find that $S$ was actually smaller: $|S'| \geq |E(U, v)| + |E(x, C^i \setminus U)| \geq |C^i| = i > j$.
	
	Now, let us be given some fixed $i, j \in \N$ with $i \neq j$, let $h$ be a natural number larger than $i$ and $j$. By previous argumentation, we know that $\OPT_\CE(G^i_B \oplus G^0_B) = \OPT_\CE(G^j_B \oplus G^0_B) = 0$, but at the same time, we know that $\OPT_\CE(G^i_B \oplus G^h_B) = i$, while $\OPT_\CE(G^j_B \oplus G^h_B) = j \neq i$. This way, there exists no constant $\Delta$ to witness gluing equality of $G^i_B$ and $G^j_B$ and we get $G^i_B \not\equiv G^j_B$.
\end{proof}

By \Cref{full:lem:BK-index} and the fact that graphs in $\Ccomplete$ have $\emptyset$ as a solution with value $0$ for both \clusterediting and \myproblem{Cluster Vertex Deletion}, it follows from \Cref{full:lem:noFII:CE_complete}:

\begin{theorem}
	The parameterized problems \clustereditingmodcvd and \clustereditingparamce do not admit boundaried kernelization.
\end{theorem}

\subsection{Maximum Cut}
\begin{lemma}\label{full:lem:noFII:MC_independent}
	\maximumcut does not have finite integer index, even on boundaried graphs with $G-B$ being an independent set.
\end{lemma}
\begin{proof}
	We fix the boundary $B$ to contain exactly two vertices, $x$ and $y$, and define an infinite series of boundaried graphs $G^i_B$ with $i \in \N$ such that $i \neq j$ leads to $G^i_B \not\equiv G^j_B$. Namely, $G^i_B$ consists of $x$, $y$ and an independent set of size $i$, in which all vertices are adjacent to both $x$ and $y$.
	In order to show that these graphs are not gluing-equivalent with respect to \maximumcut, we also define the graph $H^i_B$, which consists of a complete bipartite graph with $i$-vertices on each side. In one of the sides, every vertex is adjacent to $x$, and in the other one, every vertex is adjacent to $y$.
	
	Note that for any fixed $i \in \N$, the graph $G^i$ is constructed in such a way, that its whole edge set is a cut of size $2i$ between the boundary and non-boundary vertices. This leads to $\OPT_\MC(G^i_B \oplus H^0_B) = 2i$, as $G^i_B \oplus H^0_B = G^i$.
	
	On the other hand, for $i < j$ it holds that $\OPT_\MC(G^i_B \oplus H^j_B) = j^2 + 2j + i$: With $V \coloneq V(G^i) \setminus B$ and $U \coloneq N(x) \cap (V(H^j) \setminus B)$ and $W \coloneq N(y) \cap (V(H^j) \setminus B)$ we get a corresponding solution by using the cut between $U \cup \{y\} \cup V$ and $W \cup \{x\}$ of size $|U| \cdot |W| + |U| + |W| + |V| = j^2 + 2j + i$. Let us argument why no better solution can be found. As we soon will show, to find a cut at least as large, we would still need $U$ entirely  in one half and $W$ in the other, but then we for sure lose the edges between $x$ and either $U$ or $W$, hence we still get a solution of size at most  $j^2 + 2j + i$. Say, the hypothetical solution is a cut between vertex sets $X$ and $Y$. Let $U_X = U \cap X$ and define $U_Y, W_X$ and $W_Y$ analogically. Let $a = |U_X|$ and $b = |W_X|$. Since the total number of edges in $G \coloneq G^i_B \oplus H^j_B$ is $2i + 2j + j^2$, and since edges between $U_X$ and $W_X$ (respectively, between $U_Y$ and $W_Y$, of which there are $j-a$ and $j-b$) are not part of the cut, it holds that $|E_G(X, Y)| \leq j^2 + 2j +2i - (ab + (j-a)(j-b))$. Let $f(a, b) = ab+(j-a)(j-b)$. In order to obtain $|E_G(X, Y)| \geq j^2 + 2j + i$, it thus must hold that $f(a, b) \leq i$. We do a case distinction: If $a = 0$ and $b = j$ or the other way around, i.e., it holds that $U$ and $W$ are completely contained in the opposing sides, we get $f(a, b) = 0$. If $a$ and $b$ are both zero or both $t$, then we have $f(a,b) = j^2 + 0 > i$. Else assume w.l.o.g.\ that it holds $1 \leq a \leq j-1$, which then implies that $ab \geq b$ and $(j-a)(j-b) \geq j-b$, hence $f(a, b) \geq b + j -b = j > i$.
	
	To conclude this proof, assume we are given arbitrary, unequal $i, j, h \in \N$ and let $h$ be larger than both $i$ and $j$. Then $h$ witnesses that $G^i_B$ and $G^j_B$ are not gluing equivalent with respect to \maximumcut, as we can see by previous argumentation, that $\OPT_\MC(G^i_B \oplus H^0_B) - \OPT_\MC(G^j_B \oplus H^0_B) = 2i - 2j$ is not equal to $\OPT_\MC(G^i_B \oplus H^h_B) - \OPT_\MC(G^j_B \oplus H^h_B) = i-j$.
\end{proof}

From \Cref{full:lem:noFII:MC_independent} and \Cref{full:lem:BK-index} directly follows:

\begin{theorem}
	The parameterized problem \maximumcutmodvc does not admit boundaried kernelization.
\end{theorem}

\subsection{Tree Deletion Set}
\begin{lemma}\label{full:lem:noFII:TDS_is_B}
	\treedeletion does not  have finite integer index, even on boundaried graphs with $G-B$ being an independent set.
\end{lemma}
\begin{proof}
	Choosing our boundary $B$ to consist  of exactly one vertex $x$, we show that the equivalence relation $\equiv_{\TDS, B}^{\graphclass{independent}_B}$, from now on denoted simply $\equiv$ in this proof, has infinitely many equivalence classes. For this, we give an infinite set of boundaried graphs $G^i_B, i \in \N_{\geq 1}$ with $G^i_B \in \graphclass{independent}_B$, such that for any distinct $i, j \in \N$ it holds that $G^i_B \not\equiv G^j_B$. We define $G^i$ to be an $i$-star with $x$ as the center, i.e., $G^i$ is a graph containing the boundary vertex $x$ and $i$ other vertices $v_1, \dots, v_i$, which are all adjacent to $x$.
	In order to show non-equivalence of $G^i_B$ and $G^j_B$, we define additional boundaried graphs $H^h_B, h \in \N_{\geq 1}$, each of which consists of vertices $r$, $a_n, b_n$ with $n \in [h]$ and the boundary vertex $x$. For each $n \in [h]$ the vertices $a_n$ and $b_n$ are adjacent to each other and $x$; furthermore, $a_n$ is adjacent to $r$.
	
	With those structures defined, fix any $i, j \in \N_{\geq 1}$ with $i < j$, and let  $h = j+1$. First, note that $\OPT_\TDS(G^i_B \oplus H^1_B) = \OPT_\TDS(G^j_B \oplus H^1_B) = 1$, since both resulting graphs contain exactly one cycle which is induced by $\{a_1, b_1, x\}$, which can be hit by removing $b_1$, while preserving connectivity.
	On the other hand, we show that $\OPT_\TDS(G^i_B \oplus H^h_B) \leq i + 1$ while $\OPT_\TDS(G^j_B \oplus H^h_B) \geq h = j+1 > i+1$, which then proves our point, as we do not have the same difference in optimum values after gluing with $H^1_B$ and $H^h_B$.
	
	Note that each $F^i \coloneq G^i_B \oplus H^h_B$ and $F^j \coloneq G^j_B \oplus H^h_B$ contain $h$-many cycles, which pairwise all intersect exactly at $x$. This constellation is called an $x$-flower of order $h$, and it is easy to see, that in order to hit those $h$ cycles with less than $h$ vertices, one needs to take $x$. Both $F^i - x$ and $F^j - x$ are disconnected, but at the same time $T \coloneq \{r\} \cup \{a_n, b_n \mid n \in [h]\}$ induces a tree. Since this tree contains more than $h$ vertices, we check whether removing all the other components leads to a solution smaller than $h$. In case of $F^i-x$, it is exactly the vertices $v_1, \dots, v_i$ which are not  contained in $T$, and whose removal gives us a feasible solution of size $i+1$. However, in $F^j-x$, outside of $T$ there are $j$-many vertices $v_1, \dots, v_j$. Removing those would give us size $j+1  = h$. Thus, $F^j$ does not have a solution which is smaller than $h$.
\end{proof}

\begin{lemma}\label{full:lem:noFII:TDS_tree_B}
	\treedeletion does not have finite integer index on $\graphclass{tree}_B$.
\end{lemma}
\begin{proof}
	This proof works similarly to the proof for \Cref{full:lem:noFII:TDS_is_B}. We again choose $B = \{x\}$ and for $h \in \N_{\geq 1}$ we define the boundaried graph $H^h_B$ exactly the same. The boundaried graph $G^i_B, i \in \N_{\geq 2}$ is defined nearly the same, but  there is an extra vertex $y \neq x$, which is adjacent to all $v_1, \dots, v_i$.
	
	Let $i, j \in \N_{\geq 2}$ with $i < j$, and we emphasize that this time it holds that $i,j\geq 2$, but we still get an infinite amount of non-equivalent boundaried graphs. First, it holds that $\OPT(G^i_B \oplus H^1_B) = \OPT(G^j_B \oplus H^1 _B) = 2$. This time, we have multiple cycles, which each contain at least either $y$ or $b_1$. Removing these two vertices does not  interrupt connectedness of the rest, so in both cases the optimum is at most $2$. However, we cannot find any better solution: in search of one we would need to take $x$, as it is the only vertex contained in all cycles. This way we get multiple components, at least one of which then also needs to be contained in the solution. Hence such a solution has at least $2$ vertices.
	
	Now with $h = j+2$, we show that $\OPT(G^i_B \oplus H^h_B) \leq i+2$, while $\OPT(G^ j_B \oplus H^h_B) \geq h = j+2 > i+2$. Let $F^i = G^i_B \oplus H^h_B$ and $F^j = G^ j_B \oplus H^h_B$. For $F^i$ it is sufficient to remove all vertices contained in $G^i$, i.e., we remove $x, y, v_1, \dots, v_i$, which are $i+2$ vertices. This way we are left with the tree $H^h - x$. In order to find a solution for $F^j$ of size less than $h$, we need to remove $x$, as the graph contains an $x$-flower of order at least $h$. However, $F^j - x$ contains two components, one with $2h + 1$ vertices and the other with $j+1$. Whichever component we would choose to remove, we would get size at least $j+2$. This way, $F^j$ does not have a solution of size less than $h$.
	
	Altogether, we get $\OPT(G^j_B \oplus H^1 _B) - \OPT(G^i_B \oplus H^1_B) = 0$, which is clearly not equal to $\OPT(G^j_B \oplus H^h_B) - \OPT(G^i_B \oplus H^h_B) \geq 1$. Hence it holds that $G^i_B \not\equiv G^j_B$.
\end{proof}

By \Cref{full:lem:BK-index} in combination with \Cref{full:lem:noFII:TDS_is_B,full:lem:noFII:TDS_tree_B} we get the following theorem.

\begin{theorem}
	The parameterized problems \treedeletionvc and \treedeletiontds do not admit boundaried kernelization.
\end{theorem}

\subsection{Long Cycle and Long Path}
\begin{lemma}\label{full:lem:noFII:LC_d2}
	\longcycle does not have finite integer index, even on the class of $B$-boundaried graphs where each $v \in V(G) \setminus B$ has degree two.
\end{lemma}
\begin{proof}
	Fix $B = \{x_1, x_2, x_3, x_4\}$ and define for each $i \in \N$ the graph $G^i_B$ with additional vertices $a$ and $b_1, \dots, b_{i+1}$ and add edges such that $(x_1, a, x_2)$ forms a path of length $2$ and $(x_3, b_1, \dots, b_{i+1}, x_4)$ forms a path of length $i+2$. Further define the following two graphs: $L_B$ with additional vertex $l$ such that $(x_1, l, x_2)$ forms a path of length $2$; $R_B$ with additional vertex $r$ such that $(x_3, r, x_4)$ forms a path of length $2$. It is easy to see that for any distinct $i,j \in \N$ it holds that $\OPT(G^i_B \oplus L_B) = 4 = \OPT(G^j_B \oplus L_B)$, while $\OPT(G^i_B \oplus R_B) = i+4 \neq j+4 = \OPT(G^j_B \oplus R_B)$, implying that $G^i_B \not\equiv G^j_B$.
\end{proof}

\begin{lemma}\label{full:lem:noFII:LP_d2}
	\longpath does not have finite integer index, even on the class of $B$-boundaried graphs where each $v \in V(G) \setminus B$ has degree two.
\end{lemma}
\begin{proof}
	Fix $B = \{x_1, x_2, x_3\}$ and define for each $i \in \N$ the graph $G^i_B$ with additional vertices $b_1, \dots, b_{i+1}$ and edges such that $(x_2, b_1, \dots, b_{i+1}, x_3)$ forms a path of length $i+2$. Further let $H^i_B$ have additional vertices $c_1, \dots, c_i$ such that $(x_1, c_1, \dots, c_i)$ forms a path of length $i$. Now given distinct $i, j \in \N$ (and assume without loss of generality that $i > j$) and $h = i+3$ it is easy to see that $\OPT(G^i_B \oplus H^h_B) = h = \OPT(G^j_B \oplus H^h_B)$, while $\OPT(G^i_B \oplus H^0_B) = i + 2 > j + 2 = \OPT(G^j_B \oplus H^0_B)$, implying that $G^i_B \not\equiv G^j_B$.
\end{proof}

By \Cref{full:lem:BK-index} and \Cref{full:lem:noFII:LC_d2,full:lem:noFII:LP_d2} we get:

\begin{theorem}
	\longcycledeg and \longpathdeg do not admit polynomial boundaried kernelization.
\end{theorem}

\subsection{Dominating Set}
One especially interesting case in our series of lower bounds for boundaried kernelization is \dominatingsetvc, as \dominatingset is known to have finite integer index (see \cite{DBLP:journals/jacm/BodlaenderFLPST16}), but by showing, that it does not have single-exponential integer index, we can still rule out any polynomial boundaried kernelization for \dominatingsetvc. This time, we cannot choose $B$ to consist just of one or two vertices, and instead need to work with $B$ having arbitrary size.

\begin{lemma}\label{full:lem:noSEII:DS_vc}
	\dominatingset does not have single-exponential integer index, even on $\graphclass{independent}_B$.
\end{lemma}
\begin{proof}
	Without loss of generality, let $B = \{x_1, \dots, x_q\}$ for some arbitrary, but fixed $q \in \N$.	We give a set $\cQ$ of $2^{\binom{q}{\lfloor q / 2 \rfloor}}$-many pairwise non-equivalent (with respect to $\equiv$) boundaried graphs. It holds that $t \coloneq \binom{q}{\lfloor q/2 \rfloor} \geq (\frac{q}{\lfloor q/2 \rfloor})^{\lfloor q/2 \rfloor} \geq 2^{\lfloor q/2 \rfloor} \in \omega(|B|^c)$ for any constant $c$. Hence, $|\cQ|$ is not single-exponential.
	Let $\cD$ be the family of all size-$\lfloor q/2 \rfloor$ subsets of $[q] = \{1, \dots, q\}$. For any $\cI \subseteq \cD$ let $G^\cI_B$ be a $B$-boundaried graph with additional vertices $y_I$ for every $I \in \cI$, each of which is adjacent to vertices $x_i$ with $i \in I$. Since the vertices $y_I$ have no edges between each other, it clearly holds that $G^\cI_B \in \graphclass{independent}_B$. We set $\cQ \coloneq \{G^\cI_B \mid \cI \subseteq \cD\}$. Observe that $|\cQ| = 2^t$, so it remains to show that different $G^\cI_B, G^\cJ_B \in \cQ$ are not gluing-equivalent. Therefore, we define for each $P \subseteq [q]$ a $B$-boundaried graph $H^P_B$ with additional vertices $z_i$ for each $i \in P$, which is adjacent only to $x_i$. Observe that after gluing any $G^\cI_B$ to $H^P_B$, there is a minimum dominating set containing $x_i$ for every $i \in P$, as it is adjacent to a vertex of degree one, which needs to be dominated anyway. 
	
	Let $\cI$ and $\cJ$ be two different subsets of $\cD$. We show that $G^\cI_B \not\equiv G^\cJ_B$. First, note that both $G^\cI_B \oplus H^{[q]}_B$ and $G^\cJ_B \oplus H^{[q]}_B$ have $B$ as an optimum solution: By previous argumentation exists an optimum solution containing all $x_i$ with $i \in [q]$, and all $y_I, I \in \cI$ are adjacent to some vertices in $B$, thus all vertices are dominated by $B$. This way we get $\OPT(G^\cJ_B \oplus H^{[q]}_B) - \OPT(G^\cI_B \oplus H^{[q]}_B) = 0$.
	
	As $\cI \neq \cJ$, we can assume without loss of generality, that exists some $I \in \cI \setminus \cJ$. Consider the two graphs $F^\cI = G^\cI_B \oplus H^{[q] \setminus I}_B$ and $F^\cJ = G^\cJ_B \oplus H^{[q] \setminus I}_B$. We show that $\OPT(F^\cJ) - \OPT(F^\cI) \neq 0$, which proves that $G^\cI_B$ and $G^\cJ_B$ are not gluing equivalent. For this, first observe that $\OPT(F^\cI) \leq \lceil q/2 \rceil + 1$, since the set $\{y_I\} \cup \{x_h \mid h \in [q] \setminus I\}$ is a feasible solution of that size: for each $h \in [q] \setminus I$ it holds that $x_h$ dominates $x_h$ and $z_h$; each $y_J$ with $J \in \cI \setminus \{I\}$ is also dominated by some $x_h$ with $h \in [q] \setminus I$, and $y_I$ together with all $x_r$ with $r \in I$ are dominated by $y_I$. However, it holds that $\OPT(F^\cJ) > \lceil q/2 \rceil + 1$. By earlier argumentation we know it is safe to assume that for each $h \in [q] \setminus I$ the vertex $x_h$ is in the optimum solution. As a next step, we would need to dominate all $x_r$ with $r \in I$ using only one vertex. However, by construction no vertex in $F^\cJ$ is adjacent to all such $x_r$.
\end{proof}

Using \Cref{full:lem:BK-index,full:lem:noSEII:DS_vc} we get the following theorem.
\begin{theorem}
	The parameterized problem \dominatingsetvc does not admit polynomial boundaried kernelization.
\end{theorem}


\section{Improved kernelization for VC[mod to constant $\td$]}\label{full:section:VC_constTD}
By applying the idea of boundaried kernelization we give an improved kernel for \vertexcovertdd. Bougeret and Sau \cite{DBLP:journals/algorithmica/BougeretS19} gave the first polynomial kernel for this problem on general graphs, which was then improved and generalized by Hols et al.~\cite{DBLP:journals/siamdm/HolsKP22}. The kernelization works in a recursive manner: If $d=1$, then apply the vertex-linear kernel for \vertexcovervc. Else, given a graph $G$ and modulator $X$ to constant treedepth $d > 1$, reduce the number of connected components in $G-X$ to at most $\mathcal{O}(|X|^{\cramped{2^{d-2}+1}})$, and put from each remaining component the root vertex of its treedepth decomposition to the modulator, reducing the task to computing a kernel for \parameterizedproblem{Vertex Cover}{mod to $\td\leq d-1$}. In total, this yields a kernel with $\mathcal{O}(|X|^{\cramped{2^{\Theta(d^2)}}})$ vertices. An issue with this approach is that the roots of all components are moved to the modulator simultaneously, leading to an unnecessarily large blow-up of the number of components in each recursive step. Instead, we take a depth-first approach: At each recursive step, after having reduced the number of components, we partition the set of remaining components into groups of size at most $|X|$, and recursively process each group individually as a boundaried graph, with $X$ and the roots of the respective components forming the boundary and modulator to treedepth $d-1$. This way, we achieve a kernel with $\Oh(|X|^{2^{d-1}})$ vertices, which comes much closer to the best known lower bound of $\Oh(|X|^{2^{d-2}+1})$ vertices \cite[Theorem 1.1]{DBLP:journals/siamdm/HolsKP22}. The correctness of this approach follows from \Cref{full:thm:VC_vc} and \Cref{full:lem:bs_cc}, a \textit{gluing}-equivalence restatement for bounding the number of connected components in $G-X$. We remind that the notion of chunks, blocking sets etc., was defined in \Cref{full:section:vc}.

Let $A$ be a bipartite graph, whose vertex set consists of the connected components of $G-B$ on one side, and of the chunks of $B$ on the other side. Add an edge between component $F$ and chunk $Z$ in $A$, if and only if $N_F(Z)$ is a blocking set of $F$, i.e., if it holds that $\conf{F}{Z} > 0$. Observe that by assumption of \Cref{full:lem:bs_cc}, we can compute $\conf{F}{Z}$ in polynomial time. Using the Hopcroft-Karp algorithm for maximum matching, find a set $\cX' \subseteq \cX$ such that $|N_A(\cX')| < |\cX'|$ and a matching $M$ in $A$ that saturates $\hat{\cX} = \cX \setminus \cX'$. Note that if there exists an $\cX$-saturating matching for $A$, then it holds that $\cX' = \emptyset$ and $\hat{\cX} = \cX$. Denote by $\mathcal{F}_D$ the set of connected components of $G-B$ that are neither in the neighborhood of $\cX'$, nor contained in a matching edge of $M$.

\begin{lemma}[Analogue of {\cite[Lemma 3.4]{DBLP:journals/siamdm/HolsKP22}}]\label{full:lem:vctdd_hitchunk}
	For every $H_B$ there exists an optimum vertex cover $S$ of $G_B \oplus H_B$ with $S \cap Z \neq \emptyset$ for all $Z \in \hat{\cX}$.
\end{lemma}
\begin{proof}
	Fix any boundaried graph $H_B$ to be glued to $G_B$.
	We will go along the proof of Lemma 3.4 in the paper of Hols et al.
	
	Let $S$ be an optimum vertex cover of $G_B \oplus H_B$. If $S \cap Z \neq \emptyset$ for all $Z \in \hat{\cX}$, then we are done. Thus, assume that there exists at least one set $Z \in \hat{X}$ such that $S \cap Z = \emptyset$. Let $\tilde{\cX} = \{Z \in \hat{\cX} \mid S \cap Z = \emptyset\}$ be the set of chunks with $S \cap Z = \emptyset$, and let $\tilde{\cF} = \{H \in \cF \mid \exists Z \in \tilde{\cX}: \{Z, H\} \in M \}$ be the set of connected components of $G-B$, that are matched to a vertex in $\tilde{X}$ by $M$. Observe that $|\tilde{X}| = |\tilde{\cF}|$.
	Hols et al.\ have shown that for every $F \in \tilde{\cF}$ it holds that $|V(F) \cap S| > \OPT(F)$.
	
	Now we construct an optimum vertex cover $S'$ of $G_B \oplus H_B$ which fulfills the properties of the lemma. First, we add from each chunk $Z \in \tilde{\cX}$ one arbitrary vertex to the set $S$. We denote the resulting set $\hat{S}$. It holds that $|\hat{S}| \leq |S| + |\tilde{X}|$. Furthermore, Hols et al.\ have shown that for every $F \in \tilde{\cF}$ there exists an optimum vertex cover $S_F$ of $F$ which contains the set $Y_F = N(B \setminus \hat{S}) \cap V(F)$. This is exactly what we do in our next step: we replace for every connected component $F \in \tilde{\cF}$ the set $S \cap V(F)$ in $\hat{S}$ by the optimum vertex cover $S_F$ in $F$, that contains the set $Y_F$. Denote the resulting set by $S'$. As we have seen earlier, that for $S$ and every $F \in \tilde{\cF}$, the set $V(F) \cap S$ was larger than an optimum vertex cover for $F$, and we have replaced this part by an optimum vertex cover for $F$, it holds that $|S'| \leq |\hat{S}| - |\tilde{\cF}| \leq |S| + |\tilde{\cX}| - |\tilde{\cF}| = |S|$.
	
	It remains to show that $S'$ is a vertex cover of $G_B \oplus H_B$. Since $S'$ contains all vertices from $S \cap V(H)$, we know that every edge inside of $H$ is still covered by $S'$. Every edge inside of $(G_B \oplus H_B)[B]$ is covered by $S'$, since $S' \cap B$ is a superset of $S \cap B$. And at last, for every component $F \in \tilde{\cF}$, the set $S'$ contains all vertices of $S_H$, which covers all edges inside of $F$, and contains the neighborhood of all $B$-vertices that are not contained in $S'$.
\end{proof}

\begin{redrule}[Analogue of {\cite[Reduction Rule 3.3]{DBLP:journals/siamdm/HolsKP22}}]\label{full:rr:vctdd}
	Delete all connected components in $\mathcal{F}_D$ from $G$ and increase $\Delta$ by $\OPT(\cF_D)$.
\end{redrule}
\begin{proof}[Proof of gluing safeness]
	Fix any boundaried graph $H_B$ to glue with $G_B$.  Let $G'_B$ be the reduced instance, i.e., $G' = G - \cF_D$. It obviously holds that $\OPT(G'_B \oplus H_B) \leq \OPT(G_B \oplus H_B) - \OPT(\cF_D)$. So let us show, that the other direction also holds, i.e., that $\OPT(G_B \oplus H_B) \leq \OPT(G'_B \oplus H_B) + \OPT(\cF_D)$. Note that with the same choice of $\hat{\cX}$ we can also use \Cref{full:lem:vctdd_hitchunk} on $G'_B$. So let $S'$ be a minimum vertex cover for $G'_B \oplus H_B$ with $S' \cap Z \neq \emptyset$ for all $Z \in \hat{\cX}$. Note that in $A$, every connected component $F \in \cF_D$ is only adjacent to vertices in $\hat{\cX}$. Since every set $Z \in \hat{\cX}$ has a nonempty intersection with the set $S'$, it holds that there exists an optimum vertex cover $S_F$ of $F$ which contains the set $N_G(B \setminus S') \cap V(F)$ (similar to how it is the case for every $F \in \tilde{\cF}$ in the proof of \Cref{full:lem:vctdd_hitchunk}). Let $S$ be the set that results from adding for each component $F \in \cF_D$ the optimum vertex cover $S_F$ to the set $S'$. By construction of $S$, it is a superset of $S'$ and thus covers all edges in $G'_B \oplus H_B$. Additionally, all the new edges inside of any component $F \in \cF_D$ and between $F$ and $B \setminus S$ are covered by $S_F \subseteq S$. Thus, $S$ is a vertex cover of $G_B \oplus H_B$ of size $|S'| + \OPT(\cF_D)$.
\end{proof}

As \Cref{full:rr:vctdd} works exactly the same as \cite[Reduction Rule 3.3]{DBLP:journals/siamdm/HolsKP22}, we also refer to their proof (see \cite[Lemma 3.10]{DBLP:journals/siamdm/HolsKP22}), that \Cref{full:rr:vctdd} can be applied in polynomial time. Similarly, we refer to Hols et  al.\ for the bound of at most $\Oh(|B|^{b_\cC})$ connected components after having applied \Cref{full:rr:vctdd} (\cite[Theorem 3.8]{DBLP:journals/siamdm/HolsKP22}). 

\begin{lemma}[Analogue of {\cite[Theorem 1.3]{DBLP:journals/siamdm/HolsKP22}}]\label{full:lem:bs_cc}
	Let $\cC$ be a hereditary graph class with minimal blocking set size $b$ on which \vertexcover can be solved in polynomial time. There is a polynomial-time algorithm that given boundaried graph $G_B \in \cC_B$ returns a gluing-equivalent graph $G'_B \in \cC_B$ with the respective offset $\Delta$, such that $G'-B$ has at most $\Oh(|B|^b)$ connected components. 
\end{lemma}

\begin{theorem}\label{full:theorem:PBK:VC_modtdd}
	For any constant $d \in \N_{\geq 1}$, the parameterized problem \vertexcovertdd admits a polynomial boundaried kernelization with $\Oh((|B|+k)^{2^{d-1}})$ vertices.
\end{theorem}
\begin{proof}
	By \Cref{full:lem:Bismod}, we assume that $G-B$ has treedepth $d$ and show a boundaried kernelization with $\Oh(|B|^{2^{d-1}})$ vertices for this case. We do our proof by induction over $d$. In the base case it holds that $d = 1$, which means that $G-B$ is an independent set. In this case we can use our boundaried kernelization for \vertexcovervc from \Cref{full:section:vc} of vertex-size $\Oh(|B|)$.
	
	Now let $d > 1$ and assume that \parameterizedproblem{Vertex Cover}{mod to $\td\leq d-1$} with $B$ being also the modulator, has a boundaried kernelization of vertex-size $\Oh(|B|^{2^{d-2}})$. We use \Cref{full:lem:bs_cc} and the fact that $\Ctdd$ has minimal blocking set size at most $2^{d-2}+1$ (see \cite[Theorem 1.4]{DBLP:journals/siamdm/HolsKP22}) in order to find a gluing-equivalent boundaried graph $G'_B$ with $\Oh(|B|^{2^{d-2}+1})$ connected components. We partition these components into $t \in \Oh(|B|^{2^{d-2}+1})/|B|) = \Oh(|B|^{2^{d-2}})$ sets $Z_1, \dots, Z_t$ by packing into each $Z_i$, $i\in [t]$ all the vertices of at most $|B|$ components. Let $R_1, \dots, R_t$ be the sets of roots of the treedepth decompositions of $G^1 \coloneq G'[Z_1], \dots, G^t \coloneq G'[Z_t]$. Note that for each $i \in [t]$ it holds that $|R_i| \leq |B|$ and that $G' = G^1_B \oplus \dots \oplus G^t_B$. For each $i \in [t]$ it holds that $\td(G^i - (B \cup R_i)) \leq d-1$, and thus by induction hypothesis, we can compute in polynomial time a graph $\hat{G}^i$ of vertex size $\Oh((|B|+|R_i|)^{2^{d-2}}) = \Oh(|B|^{2^{d-2}})$ and constant $\Delta_i$ with $\OPT_\VC(G^i_{B \cup R_i} \oplus H_{B \cup R_i}) =  \OPT_\VC(\hat{G}^i_{B \cup R_i} \oplus H_{B \cup R_i}) + \Delta_i$ for any $H_{B \cup R_i}$, and thus by \Cref{full:lem:BD} also $G^i_B \equiv_{\VC, B} \hat{G}^i_B$ with same offset. By repeated use of \Cref{full:lem:glue_equiv_get_equiv}, it holds that $G'_B \equiv_{\VC, B} \hat{G}_B$ with $\hat{G} = \hat{G}^1_B \oplus \dots \oplus \hat{G}^t_B$ and offset $\sum_{i = 1}^{t} \Delta_i$. Thus, we can output the boundaried graph $\hat{G}_B$ which has vertex-size $t \cdot \Oh(|B|^{2^{d-2}}) = \Oh(|B|^{2^{d-1}})$.
\end{proof}

By \Cref{full:lem:PBK-PK} this also gives the regular polynomial kernel.

{
	\renewcommand{\thetheorem}{\ref*{full:thm:vc_tdd}}
	\begin{theorem}
		For any constant $d \in \N_{\geq 1}$, the parameterized problem \vertexcovertdd admits a polynomial kernelization with $\Oh(k^{2^{d-1}})$ vertices.
	\end{theorem}
	\addtocounter{theorem}{-1}
}


\section{Conclusion}\label{full:section:conclusion}

Boundaried kernelization is a model for efficient local preprocessing, inspired by protrusions and protrusion replacement used in meta kernelization. It allows for provably effective preprocessing when only local parts of the input exhibit useful structure, while this need not be true for the entire input. We have showed several polynomial-sized boundaried kernelizations as well as a number of unconditional lower bounds ruling out such preprocessing for other problems. This is unlike the typical use of protrusions where one can appeal to finite integer index to obtain a ``small'' replacement, which will usually be of size exponential in treewidth and boundary size (but still constant). That being said, finite (integer) index is still a prerequisite for having a boundaried kernelization of any size. Similarly, a (polynomial) kernelization is required for a problem to admit a (polynomial) boundaried kernelization.

While the present definition of boundaried kernelization is restricted to pure graph problems (parameterized by some graph optimization problem), there is no reason not to pursue such local kernelization also in more general settings. Rather, it was appealing to keep the notation (somewhat) concise by relying on the robust setting of boundaried graphs and the established notion of finite (integer) index. Moreover, even in the restricted setting of pure graph problems, we were able to showcase a number of different behaviors of target problems. Last but not least, this allowed us to define what is a boundaried kernelization for a given (parameterized) problem, rather than having to first define a specific boundaried version of each problem.

That being said, to go beyond pure graph problems necessitates clean definitions for what a local part of an instance should be. This is probably best formalized by having a problem-specific gluing operator on the level of entire instances. (As a simple example, consider gluing two boundaried instances \steinertree, where we should glue the graphs but also take the union of the respective terminal sets.) Ideally, such an operator has some form of neutral element/instance so that the connection to (regular) kernelization can be retained. As future work, we aim to extend boundaried kernelization (at least) to problems on general finite structures, which should be sufficient for a large(r) range of target problems while still enabling a generic definition without problem-specific parts.

Further directions of future work are, obviously, to research further what problems with polynomial kernelization also admit polynomial boundaried kernelizations (for pure graph problems and beyond), and to establish unconditional lower bounds where this is infeasible. It is also interesting to perhaps identify some general properties of problems that admit a polynomial kernelization but no polynomial boundaried kernelization, like, e.g., \treedeletion.

Finally, we showed that our polynomial boundaried kernelization for \vertexcovervc can be used to get a polynomial kernelization for \vertexcovertdd, and in fact got a significant improvement over the previous kernelization. Can this be further improved to match the lower bound via maximum size of minimal blocking sets~\cite{DBLP:journals/siamdm/HolsKP22}? Can one find more applications of polynomial boundaried kernelization for obtaining polynomial kernelizations for more complex parameterizations?

\bibliographystyle{abbrv}
\bibliography{bk}


\appendix

\section{Problem definitions}\label{app:problems}
\begin{center}
	\fbox{
		\parbox{.95\textwidth}{
			\vertexcover (\myproblem{vc})
			\begin{itemize}
				\item \textbf{Type:} Minimization problem.
				\item \textbf{Input:} A graph $G$.
				\item \textbf{Feasible solutions:} A vertex set $S \subseteq V(G)$ s.t.\ $G-S$ is an independent set.
				\item \textbf{Solution value:} The size $|S|$ of $S$.
			\end{itemize}
}}\end{center}

\begin{center}
	\fbox{
		\parbox{.95\textwidth}{
			\feedbackvertexset (\myproblem{fvs})
			\begin{itemize}
				\item \textbf{Type:} Minimization problem.
				\item \textbf{Input:} A graph $G$.
				\item \textbf{Feasible solutions:} A vertex set $S \subseteq V(G)$ s.t.\ $G-S$ is a forest.
				\item \textbf{Solution value:} The size $|S|$ of $S$.
			\end{itemize}
}}\end{center}

\begin{center}
	\fbox{
		\parbox{.95\textwidth}{
			\treedeletion (\myproblem{tds})
			\begin{itemize}
				\item \textbf{Type:} Minimization problem.
				\item \textbf{Input:} A graph $G$.
				\item \textbf{Feasible solutions:} A vertex set $S \subseteq V(G)$ s.t.\ $G-S$ is a tree.
				\item \textbf{Solution value:} The size $|S|$ of $S$.
			\end{itemize}
}}\end{center}

\begin{center}
	\fbox{
		\parbox{.95\textwidth}{
			\dominatingset (\myproblem{ds})
			\begin{itemize}
				\item \textbf{Type:} Minimization problem.
				\item \textbf{Input:} A graph $G$.
				\item \textbf{Feasible solutions:} A vertex set $S \subseteq V(G)$ s.t.\ $N_G[S] = V(G)$.
				\item \textbf{Solution value:} The size $|S|$ of $S$.
			\end{itemize}
}}\end{center}

\begin{center}
	\fbox{
		\parbox{.95\textwidth}{
			\myproblem{Long Cycle} (\myproblem{lc})
			\begin{itemize}
				\item \textbf{Type:} Maximization problem.
				\item \textbf{Input:} A graph $G$.
				\item \textbf{Feasible solutions:} A cycle $C = (v_1, v_2, \dots, v_r, v_1)$ in $G$ s.t.\ $v_1, \dots, v_r$ are distinct vertices from $G$.
				\item \textbf{Solution value:} The length $r$ of $C$.
			\end{itemize}
}}\end{center}

\begin{center}
	\fbox{
		\parbox{.95\textwidth}{
			\myproblem{Long Path} (\myproblem{lp})
			\begin{itemize}
				\item \textbf{Type:} Maximization problem.
				\item \textbf{Input:} A graph $G$.
				\item \textbf{Feasible solutions:} A path $P = (v_1, v_2, \dots, v_r)$ in $G$ s.t.\ $v_1, \dots, v_r$ are distinct vertices from $G$.
				\item \textbf{Solution value:} The length $r-1$ of $P$.
			\end{itemize}
}}\end{center}

\begin{center}
	\fbox{
		\parbox{.95\textwidth}{
			\myproblem{Hamiltonian Cycle} (\myproblem{hc})
			\begin{itemize}
				\item \textbf{Type:} Decision problem.
				\item \textbf{Input:} A graph $G$.
				\item \textbf{Question:} Does there exist a cycle $C$ in $G$ that visits all vertices of $G$?
			\end{itemize}
}}\end{center}

\begin{center}
	\fbox{
		\parbox{.95\textwidth}{
			\myproblem{Hamiltonian Path} (\myproblem{hp})
			\begin{itemize}
				\item \textbf{Type:} Decision problem.
				\item \textbf{Input:} A graph $G$.
				\item \textbf{Question:} Does there exist a path $P$ in $G$ that visits all vertices of $G$?
			\end{itemize}
}}\end{center}

\begin{center}
	\fbox{
		\parbox{.95\textwidth}{
			\myproblem{Cluster Editing} (\myproblem{ce})
			\begin{itemize}
				\item \textbf{Type:} Minimization problem.
				\item \textbf{Input:} A graph $G$.
				\item \textbf{Feasible solutions:} An edge set $E \subseteq \{\{u, v\} \mid u,v \in V\}$ s.t.\ $G \triangle E \coloneq (V(G), ((E(G) \cup E) \setminus (E(G) \cap E)))$ is a cluster.
				\item \textbf{Solution value:} The size $|E|$ of $E$.
			\end{itemize}
}}\end{center}

\begin{center}
	\fbox{
		\parbox{.95\textwidth}{
			\myproblem{Cluster Vertex Deletion} (\myproblem{cvd})
			\begin{itemize}
				\item \textbf{Type:} Minimization problem.
				\item \textbf{Input:} A graph $G$.
				\item \textbf{Feasible solutions:} A vertex set $S \subseteq V(G)$ s.t.\ $G-S$ is a cluster.
				\item \textbf{Solution value:} The size $|S|$ of $S$.
			\end{itemize}
}}\end{center}

\begin{center}
	\fbox{
		\parbox{.95\textwidth}{
			\maximumcut (\myproblem{mc})
			\begin{itemize}
				\item \textbf{Type:} Maximization problem.
				\item \textbf{Input:} A graph $G$.
				\item \textbf{Feasible solutions:} A partition $X \dot\cup Y$ of $V(G)$.
				\item \textbf{Solution value:} The size $|E|$ of $E = \{\{u, v\} \in E(G) \mid u \in X, v \in Y\}$.
			\end{itemize}
}}\end{center}

\begin{center}
	\fbox{
		\parbox{.95\textwidth}{
			\myproblem{\myproblem{Number Of Vertices With Degree $\neq 2$}} ($\#v, \text{deg}(v) \neq 2$)
			\begin{itemize}
				\item \textbf{Type:} Minimization problem.
				\item \textbf{Input:} A graph $G$.
				\item \textbf{Feasible solutions:} A vertex set $S \subseteq V(G)$ s.t.\ all vertices in $V(G) \setminus S$ have degree equal two.
				\item \textbf{Solution value:} The size $|S|$ of $S$.
			\end{itemize}
		}
	}
\end{center}
\end{document}